\documentclass[pra,twocolumn,preprintnumbers,superscriptaddress,amsmath,amssymb]{revtex4-1}
\usepackage{graphicx}
\usepackage{mathrsfs}
\usepackage{amsfonts}
\usepackage{times}
\usepackage{amsmath}
\usepackage{amsthm}
\usepackage{leftidx}
\usepackage{tikz}
\usepackage{tikz-network}
\usepackage{color}
\usepackage{subcaption}
\usepackage{placeins}
\usepackage[colorlinks,linkcolor=blue,citecolor=blue]{hyperref}
\usepackage{cleveref}
\DeclareMathOperator*{\argmax}{argmax}

\newtheorem{definition}{Definition}
\newtheorem{theorem}{Theorem}

\usepackage{mathtools}
\usepackage{subcaption}
\usepackage{graphicx}

\newtheorem{assumption}{Assumption}
\newtheorem{proposition}{Proposition}

\usepackage{algorithm}
\usepackage{algpseudocode}

\DeclareRobustCommand{\orcidicon}{%
	\begin{tikzpicture}
	\draw[lime, fill=lime] (0,0) 
	circle [radius=0.16] 
	node[white] {{\fontfamily{qag}\selectfont \tiny ID}};
	\draw[white, fill=white] (-0.0625,0.095) 
	circle [radius=0.007];
	\end{tikzpicture}
	\hspace{-2mm}
}
\foreach \x in {A, ..., Z}{%
	\expandafter\xdef\csname orcid\x\endcsname{\noexpand\href{https://orcid.org/\csname orcidauthor\x\endcsname}{\noexpand\orcidicon}}
}
\usepackage{caption}
\usepackage{subcaption}

\begin{document}
\title{A Mutual Information-based Metric for Temporal Expressivity and Trainability Estimation \\ in Quantum Policy Gradient Pipelines}
\author{Jaehun Jeong\orcidA{}}
\email{abin1125@snu.ac.kr}
\affiliation{College of Liberal Studies, Seoul National University, Seoul 08826, South Korea}
\affiliation{Team QST, Seoul 08826, South Korea}
\author{Donghwa Ji\orcidB{}}
\email{donghwa722@gmail.com}
\affiliation{College of Liberal Studies, Seoul National University, Seoul 08826, South Korea}
\affiliation{Team QST, Seoul 08826, South Korea}
\author{Kabgyun Jeong\orcidC{}}
\email{kgjeong6@snu.ac.kr}
\affiliation{Institute of Computer Technology, College of Engineering, Seoul National University, Seoul, 08826, Korea}
\affiliation{Team QST, Seoul 08826, South Korea}

\date{\today}

\begin{abstract}
In recent years, various limitations of conventional supervised learning have been identified, motivating the development of reinforcement learning--and quantum reinforcement learning that leverages quantum resources such as entanglement and superposition. Among the various reinforcement learning methodologies, the policy gradient method is considered to have many benefits; for instance, it allows an agent to learn without explicitly knowing the crucial information of the environment such as state transition probabilities and initial state distribution. Meanwhile, from the perspective of learning, two indicators are often regarded as significant: expressivity and trainability (for gradient-based methods). While a number of attempts have been made to quantify the expressivity and trainability of Neural Network models and PQCs, clear efforts suitable for reinforcement learning settings have so far been lacking, despite the inherent differences between conventional supervised learning and reinforcement learning. Therefore, in this study, we propose revising the notion of expressivity into a temporal expressivity suited to reinforcement learning dynamics, and show that the mutual information between the action distribution and the discretized reward signal provides an upper bound for the scaled gradient norm, while yielding an information-theoretic decomposition and a residual-aware upper bound for the proposed temporal expressivity metric. Finally, under explicit concentration assumptions, we show that MI-TET induces an assumption-based, one-sided prescreening criterion for initialization-time gradient fragility across PQC architectures.
\end{abstract}
\maketitle

\section{Introduction} \label{sec:intro}
%\emph{Introduction.---}
Supervised Learning (SL), which is a classical and most popular machine learning paradigm, describes a situation where explicit answer labels are provided for each data samples and the model extracts some correlations between those input data and output labels so that it can infer the appropriate label even for new, unseen data. This SL paradigm has been studied extensively across various domains, especially with the advent of deep learning in the 2010s, which has achieved remarkable success in areas such as image classification and natural language processing \cite{DeepLearning2015}.

Yet, supervised learning exhibits several limitations. Key limitation is that, when attempting to train machines to handle the diverse range of problems existing in the real world, it becomes impractical to provide explicit labels for every possible problem. For instance, consider the case of training a four-legged robot that must learn to walk~\cite{Cs229}. Given the wide variety of physical environments it may encounter--such as stairs, inclined surfaces of varying slopes, or public transportation systems like busses and subways--it would be almost impossible to prescribe explicit labels in advance such as ``adjust \emph{this} component of the body following \emph{this} precise way \emph{right at that time.}'' for every scenario, especially when the infinity of possible real-world conditions is taken into account. Even if possible, such an approach would be highly inefficient. This challenge of providing explicit labels extends beyond just robotics and is a recurring issue across many problem domains.

In this respect, reinforcement learning (RL) has been regarded as a major paradigm to overcome this shortcomings of supervised learning. Instead of requiring deterministic answer labels, RL provides the agent with a reward that reflects how close its chosen action is to an ideal one within a given specific learning environment, and the agent learns accordingly. Among the various pipelines developed under this framework, policy gradient methods, which update policy functions through gradient ascent, stand out. This method is particularly attractive since it allows learning even with only limited information about the environment, and as it is one of the gradient-based methods, it can benefit from the power of existing well-devised optimizers--such as Adam.

With the growing interest in the potential advantages of quantum computing over classical computing, there has been a surge of research attempting to bring reinforcement learning into quantum computation. Policy gradient methods are no exception, as they can be easily understood as one type of variational quantum algorithm. Accordingly, one can relatively readily and intuitively implement such algorithms by designing suitable parameterized quantum circuits (PQCs). While numerous proposals have explored different designs, in this study, we employ the parameterized quantum circuit proposed by \emph{Jerbi et al.}~\cite{ReUploadingPQC}.

From the perspective of learning theory, whether classical or quantum, two indicators are widely recognized as critical: expressivity and trainability. Expressivity, which measures how well a model---or in the quantum case, a PQC---can represent or approximate a wide range of functions, is important as it is related closely to the fundamental bias-variance tradeoff. Trainability, on the other hand, concerns how well gradient-based methods avoid issues such as vanishing(or sometimes, exploding) gradients, and how reliably they converge toward optimal solutions. Since this fundamentally governs the likelihood and stability of successful learning, there is a clear need for suitable quantitative measures of trainability.

Although many efforts have been devoted to quantifying these two indicators---either independently~\cite{HEG, FIS} or jointly~\cite{NewRef}---they have exhibited limitations. In particular, they both are not suitable to capture the dynamic, time-varying nature  that arises from the exploitation-exploration tradeoff, which is central to reinforcement learning. Moreover, many approaches rely on quantifying those indicators only at initialization, for randomly selected parameter values, thereby reflecting merely the potential coverage of functions or overall learning stability, thus not appropriate to account for the inherently dynamic nature of the reinforcement learning paradigm such as currently changing policy function and resulting variability in data distributions.

To address these issues, in this work, we propose a slight different definition of \textbf{`temporal'} expressivity tailored to reinforcement learning and introduce mutual information between the action distribution and the reward-signal distribution as a proxy that can effectively track both that \emph{newly defined} expressivity and trainability over episodes. More concretely, we formalize the relationship between the gradient norm and mutual information, as well as between the \emph{novel} expressivity notion and mutual information, in the form of inequalities. Through this analysis, we demonstrate that mutual information serves as a useful and efficient indicator capable of monitoring the overall trend of both expressivity and trainability in quantum reinforcement learning settings. 

Moreover, note that we can induce a useful screening protocol over various PQC structures just before the learning starts, where we possibly select \emph{initialization-time gradient eliminating} PQC styles using the trainability theorem. In this way, our MI-TET shares some common ground with other conventional expressivity metrics; but again, note that our main focus still lies in temporal tracking.

\section{Preliminary} \label{sec:prelim}
%\emph{Preliminary.---}
\subsection{Policy Gradient Fundamentals}
We briefly introduce reinforcement learning and the policy gradient methods. The following introduction is based on the material in the introductory notes~\cite{Cs229}.

The Markov Decision Process is commonly used to describe reinforcement learning environments.
\begin{definition} [Markov Decision Process]
    Any reinforcement learning environment can be simplified to a tuple of $(S, A, \{P_{sa}\}, \gamma, R)$, which is given as follows: 
    \begin{itemize}
        \item $S$: set of states, $A$: set of actions, $\{P_{sa}\}$: state transition probabilities. \\
        - The environment shows the agent a certain state, and the agent interacts with the environment by selecting the optimal action according to that. After this, the environment changes the state stochastically depending on the prefixed probability $P_{sa}$.
        \item $\gamma$: discount factor, $\gamma \in \left[0, 1\right]$. \\
        - It controls the trade-off between immediate and future rewards; in continuing tasks it also helps ensure convergence of the total return.
        \item $R$ : reward function, $S \times A \rightarrow \mathbb{R}$. \\
        - It provides a scalar reward signal used to evaluate state--action pairs.
    \end{itemize}
    Also introduce policy function $\pi$ : $S \rightarrow \Delta (A)$, which plays a significant role in policy gradient. \\
    - It gets input of a current state, and outputs the probability distribution of the actions with that fixed state.
\end{definition}
With this mathematical framework, the primary objective in this reinforcement learning regime becomes maximizing the expected total payoff, which is
\begin{equation*}
    \mathbb{E}\left[\sum_{t = 0}^{T - 1} \gamma^t R(s_t, a_t)\right]
\end{equation*} and can be seen as our objective function, 
by fine-tuning the policy function $\pi$. Then the following question naturally arises. \emph{How we maximize this objective function?}

There are many algorithms to do this job, such as well-known value function iteration and policy iteration, but in this work, as mentioned earlier, we will focus on a gradient-based one, policy gradient \cite{Cs229,WilliamsREINFORCE1992,SuttonPolicyGradient1999}. For this, there are some assumptions required to be satisfied.
\begin{itemize}
    \item We only consider the case with finite horizon $T$, which signifies the game definitely, in a finite timestamp, ends.
    \item Our policy function is actually a randomized policy $\pi_\theta$ with parameterized by trainable $\theta$; hence the policy function is also trainable.
\end{itemize}

Moreover, we introduce some useful definitions and notations here: 
\begin{itemize}
    \item Trajectory $\tau$: $(s_0, a_0, \cdots, s_{T - 1}, a_{T - 1}, s_T)$ \\
    - It can be viewed as \emph{learning history}; $s_T$ is for the place-holder, given the state transition flow of the environment. What really matters is $(s_0, a_0, \cdots, s_{T - 1}, a_{T - 1})$.
    \item Distribution of the trajectories $\tau$ : $P_\theta(\tau)$.
    \item The total payoff with the fixed, specific trajectory $\tau$ : $f(\tau) = \sum_{t = 0}^{T - 1} \gamma^t R(s_t, a_t)$.
\end{itemize}

Then, our objective function can be rewritten by
\begin{equation}
    \eta(\theta) := \mathbb{E}_{\tau \sim P_\theta}\left[f(\tau)\right] = \int P_\theta(\tau) f(\tau) d\tau
\end{equation}
and apply the gradient ascent method to this and update $\theta$ accordingly. One observation is that
\begin{align}
    \nabla \eta(\theta) & = \int (\nabla P_\theta(\tau)) f(\tau) d\tau \\ &= \int (\nabla \log P_\theta(\tau) P_\theta(\tau)) f(\tau) d\tau \\
    & = \mathbb{E}\left[(\nabla \log P_\theta(\tau)) f(\tau)\right].
\end{align}
Next step is to calculate $\nabla \log P_\theta (\tau)$. Note that
\begin{equation}
    P_\theta(\tau) = \mu(s_0) \prod_{t=0}^{T - 1} \pi_\theta(a_t | s_t) P(s_{t + 1} | s_t, a_t),
\end{equation}
where $\mu$ is a distribution of states that gives us the probability of a certain state being chosen as the initial state. Then we can consider the log-derivatives,
\begin{align}
\log P_\theta(\tau) 
    &= \log \mu(s_0) \\
    &+ \sum_{t=0}^{T - 1} \left(\log \pi_\theta(a_t | s_t) + \log P(s_{t + 1} | s_t, a_t)\right), 
\end{align}
and therefore,
\begin{equation}
    \nabla \log P_\theta(\tau) = \sum_{t = 0}^{T - 1} \nabla \log \pi_\theta(a_t | s_t)
\end{equation}
where we used that $\mu$ and $P(s_{t + 1} | s_t, a_t)$ do not depend on $\theta$.
One interesting fact here is that we do not need explicit knowledge of the initial-state distribution $\mu$ or the transition kernel $P(s' | s, a)$. This is one reason why policy gradient methods are attractive when the environment dynamics are unknown.

Substituting the equation (6) into (3), we get
\begin{equation}
    \nabla \eta(\theta) = \mathbb{E}\left[\left(\sum_{t=0}^{T - 1} \nabla \log \pi_\theta (a_t | s_t)\right)\left(\sum_{t = 0}^{T - 1} \gamma^t R(s_t, a_t)\right)\right].
\end{equation}
Moreover, we can easily prove that
\begin{equation}
    \mathbb{E}\left[\nabla \log \pi_\theta(a_t | s_t) \cdot \sum_{j < t} \gamma^j R(s_j, a_j)\right] = 0
\end{equation}
and the equation (8) now simplifies to
\begin{equation}
    \nabla \eta(\theta) = \mathbb{E}\left[\sum_{t = 0}^{T - 1} \left(\nabla \log \pi_\theta(a_t | s_t) \cdot \left(\sum_{j = t}^{T - 1} \gamma^j R(s_j, a_j)\right)\right)\right].
\end{equation}

Let us introduce the following notations, 
\begin{equation}
    G_t := \sum_{k = 0}^{T - t - 1} \gamma^k R_{t + k}.
\end{equation}
Then,
\begin{align}
    &\mathbb{E}\left[\nabla \log \pi_\theta(a_t | s_t) \cdot G_t\right] \nonumber\\ & = \mathbb{E}\left[\mathbb{E}\left[\nabla \log \pi_\theta (a_t | s_t) \cdot G_t | s_t, a_t\right]\right] \\
    & = \mathbb{E}\left[\nabla \log \pi_\theta (a_t | s_t) \cdot \mathbb{E}\left[G_t | s_t, a_t\right]\right].
\end{align}
Now we see the form like the conventional Q function, that is,
\begin{equation}
    Q_t^\pi (s, a) = \mathbb{E}\left[\sum_{k=0}^{T - t - 1} \gamma^{k} R_{t + k} \middle| s_t = s, a_t = a\right].
\end{equation}
Thus, it again gets rephrased as
\begin{equation}
    \mathbb{E}\left[\nabla \log \pi_\theta (a_t | s_t) \cdot G_t\right] = \mathbb{E}\left[\nabla \log \pi_\theta(a_t | s_t) \cdot Q_t^{\pi_\theta} (s_t, a_t)\right].
\end{equation}

Now, 
\begin{align}
    \nabla \eta(\theta) & = \mathbb{E}\left[\sum_{t = 0}^{T - 1} \nabla \log \pi_\theta(a_t | s_t) \cdot \gamma^t G_t\right] \\
    & = \sum_{t = 0}^{T - 1} \gamma^t \mathbb{E}\left[\nabla \log \pi_\theta(a_t | s_t) \cdot G_t\right] \\
    & = \sum_{t = 0}^{T - 1} \gamma^t \mathbb{E}\left[\nabla \log \pi_\theta(a_t | s_t) \cdot Q_t^{\pi_\theta} (s_t, a_t)\right].
\end{align}
Consequently, the choice between $G_t$ and $Q_t^{\pi_\theta}$ for gradient estimation is interchangeable within the expected policy gradient term \cite{SuttonPolicyGradient1999}, and we will use the notation $Y_t \in \{G_t, Q_t^{\pi_\theta}(s_t, a_t)\}$ from now on this work.

\subsection{Quantum Policy Gradient Framework}
One highlighted approach in the Noisy Intermediate-Scale Quantum (NISQ) era is the so called hybrid quantum machine learning approach, where hypothesis family for a given learning task is realized by parameterized quantum computation, with extra classical optimization algorithm acting on it. More practically, those quantum parts are generally given as parameterized quantum circuits (PQCs), which have been proved to solve various learning tasks efficiently.

The whole quantum circuit, which is shown in Fig. 1 of~\cite{ReUploadingPQC} is characterized by trainable parameters $(\phi, \lambda, \omega)$, which are rotation angles, scaling parameters, and the weight of the Hermitian measurement operators, respectively. Note that the state is not encoded at once, rather being reuploaded multiple times during the whole circuit evaluation, with the scaling parameters $\lambda$.

And the whole quantum policy gradient algorithm we adopted from the~\cite{ReUploadingPQC} is given as 
\begin{algorithm}[H]
\caption{REINFORCE with PQC policies and value-function baselines}
\label{alg:pqc_reinforce}
\begin{algorithmic}[1]
    \State \textbf{Input:} a PQC policy $\pi_{\theta}$; a value-function approximator $\tilde{V}_{\omega}$
    \State Initialize parameters $\theta$ and $\omega$; 
    \While{True}
        \State Generate $N$ episodes $\{(s_{0},a_{0},r_{1},...,s_{H-1},a_{H-1},r_{H})\}_{i}$ following $\pi_{\theta}$
        \For{episode $i$ in batch}
            \State Compute the returns $G_{i,t} \leftarrow \sum_{t^{\prime}=1}^{H-t}\gamma^{t^{\prime}}r_{t+t^{\prime}}^{(i)}$ 
            \State Compute the gradients $\nabla_{\theta}\log \pi_{\theta}(a_{t}^{(i)}|s_{t}^{(i)})$ 
            \State Fit $\{\tilde{V}_{\omega}(s_{t}^{(i)})\}_{i,t}$ to the returns $\{G_{i,t}\}_{i,t}$; 
        \EndFor
        \State Compute $\Delta\theta \leftarrow \frac{1}{N}\sum_{i=1}^{N}\sum_{t=0}^{H-1}\nabla_{\theta}\log \pi_{\theta}(a_{t}^{(i)}|s_{t}^{(i)})(G_{i,t} - \tilde{V}_{\omega}(s_{t}^{(i)}))$ ; 
        \State Update $\theta \leftarrow \theta + \alpha\Delta\theta$; 
    \EndWhile
\end{algorithmic}
\end{algorithm}
\noindent which is quite similar to what have been introduced earlier in policy gradient statements, except for the part of the baseline. It is often used to reduce the variance of gradient estimation, and as stated in \Cref{alg:pqc_reinforce}, the value-function approximator is one of the popular choices~\cite{Cs229, baseline1}.

\subsection{Concepts of Expressivity and Trainability}
In both classical and quantum settings, two concepts are particularly important for gradient-based learning: expressivity and trainability. From an approximation-theoretic perspective, learning is defined statistically as the problem of finding the optimal approximation function $\hat{f_n}$ of the real, unknown function $f_\rho$ (e.g. true regression function), within certain function class $\mathcal{F}$ (e.g. neural networks with ReLU activation) given an observed dataset $\mathcal{D} = \{(x_i, y_i)\}_{i=1}^n$~\cite{LearningTheory}. Consider the regression model
\begin{equation}
    \mathbf{y}_i = f_\rho(\mathbf{x}_i) + \epsilon_i, \quad i = 1, 2, \cdots, n.
\end{equation}
Our goal is to find a function $f$ that minimizes the so called population risk $\mathcal{E}(f)$, which is defined as
\begin{equation}
    \mathcal{E}(f) := \mathbb{E}_{(\mathbf{x}, \mathbf{y}) \sim \rho} \left[(\mathbf{y} - f(\mathbf{x}))^2\right],
\end{equation}
where $\rho$ denotes the joint distribution of $\mathbf{x}$ and $\mathbf{y}$. However, it is hardly known the true $\rho$ in most of the cases, thus we instead consider the empirical risk minimizer $\hat{f_n}$ over the dataset $D$,
\begin{align}
    \hat{f_n} &= \text{argmin}_{f \in \mathcal{F}} \mathcal{E}_D(f)\\
    &:= \text{argmin}_{f \in \mathcal{F}} \left[\frac{1}{n} \sum_{i=1}^n (y_i - f(x_i))^2\right].
\end{align}
Now the concept of excess risk is induced, which is $\mathbb{E}_{\mathbf{X} \sim \rho_\chi} \left[(\hat{f_n}(\mathbf{X}) - f_\rho(\mathbf{X}))^2\right]$. It serves as an important statistical metric directly measuring the generalizability of our candidate function $\hat{f_n}$ in the unseen data setting. Interestingly, it is well known that the following inequality holds~\cite{approximation}: 
\begin{align}
    &\mathbb{E}_{\mathbf{X}\sim \rho_\chi} \left[(\hat{f_n}(\mathbf{X}) - f_\rho (\mathbf{X}))^2\right] \nonumber \\ & \leq \frac{\text{Complexity Measure of $\mathcal{F}$}}{n}+ \text{Approx. Error}^2
    \label{eq:my_equation}
\end{align}
The right-hand side can be seen as reflecting a trade-off between the approximation error and the complexity of the function class $\mathcal{F}$. 

This concept is not limited to regression and can be extended to other learning tasks, where it is commonly described as the bias-variance tradeoff. If the model's complexity (i.e., the complexity of $\mathcal{F}$) is too low, the model is prone to underfitting. It is characterized by high bias and highly likely to fail to capture the underlying structure of the data, often leading to poor generalization performance. If the model's complexity is too high, the model is prone to overfitting. This fits the training data very well but performs poorly on unseen data. Such a model suffers from high variance. Moreover, in contrast to classical models where test error typically follows a U-shaped curve (a single descent), modern over-parameterized deep learning models often exhibit the `double descent' phenomenon. This phenomenon reveals that increasing model complexity far beyond the classical overfitting point can lead to a second descent in test error, often achieving performance superior to the original `sweet spot'~\cite{bias-variance}. Thus, quantifying model complexity remains an important task in both classical and modern learning regimes~\cite{ExpressivePower, ExpressivePower2}.

In the literature, the term expressivity is often used without a clear distinction from complexity and capacity, since they all share the same intuitive concept~\cite{DeepLearningBook}. In reinforcement learning, the model outputs probability distribution of action and we argue the deviation of those distributions through multiple evaluations being an important metric. Hence, we use the term expressivity during this paper, which we think literally natural to capture those fundamentals.

As mentioned earlier, since the parameterized quantum circuits take the role of function class $\mathcal{F}$ theoretically, measuring the given circuit's complexity properly also remains significant in the quantum machine learning paradigm. There has been one notable attempt to tackle this issue, shown in the~\cite{HEG}. In this study, the author measures expressivity by the deviation between the output state's fidelity distribution of the target PQC and Haar random circuit. Even though these attempts might work well for the conventional supervised learning tasks, we conceived it is not suitable for reinforcement learning environments, since it just measures the `static', `fixed' models' complexity through multiple times of sampling, while what is really important in RL is the actual `temporal, time dependent fluctuation' of the policy function's expressivity considering the fundamentals of the nonstationary property in RL~\cite{ReUploadingPQC}; this discussion will be further described in the main results section.

Meanwhile, when using gradient descent/ascent, the issue of trainability also becomes one dominant problem. This informal concept of trainability, signifies literally \emph{how well the optimization goal is achieved without any issues such as gradient vanishing/exploding, too much oscillation or even the divergence, which is actually the failure of given learning task}. There have been some papers considering these issues, for instance gradient vanishing or exploding in the classical deep learning models, since it also utilizes gradient information for updating the parameter~\cite{ClassicalTrainability}. With this trainability concept in mind, one can view the Adam optimizer~\cite{Adam}, as the one of the most hopeful optimizers proposed to solve this trainability issue, as well.

In the quantum computation regime, one can easily observe the phenomena of `Barren Plateau', which is a quantum counterpart of the gradient vanishing problem~\cite{BarrenPlateaus}. In that work, the authors examined the empirical results that random circuits might be an unsuitable choice for moderately many qubits circuit while also claiming that the gradient might be concentrated to zero, which is non-utilizable, with exponentially probability decay with the number of qubits. Thus, considering the number of qubits is a fundamental parameter defining the given quantum circuit's compelxity, one can easily find that the expressivity and trainability is correlated to the certain, significant extent; and that motivates us to develop one novel metric which can track both of them simultaneously.

A nice reference for this trainability notion is~\cite{FIS}. The authors utilize the fisher information matrix spectrum, which shows how the eigenvalues of the fisher information matrix are distributed. Intuitively, one can expect the optimization performance of the model which has more concentrated eigenvalues to the zero would be poor and vice versa, since the former might signify more `flat' curvature of the loss landscape; and this is well discussed in the Figure 2 and 3 of that paper. Interestingly, this intuition holds true thus bridging the convergence speed and the fisher information spectrum in a trainability manner. However, this approach (or metric) for trainability is still not feasible for the reinforcement learning environment since it calculates `empirical' fisher information matrix with `randomly' sampled parameter sets and data sets at first glance, again based on the model's static and fixed structure. What is rather important in reinforcement learning is how the gradient or its variance gets concentrated to zero during the whole learning of the agent, which can be viewed to be more `temporal'; and that motivates us to introduce the concept of MI-TET, Mutual Information based Temporal Expressivity and Trainability measure. \\

\section{Main Results}
\subsection{Formulation of MI-TET}
First of all, let us give the formal definition of \textbf{MI-TET}, \emph{Mutual Information-based Temporal Expressivity and Trainability measure}. Before defining MI-TET, we distinguish two notions of time used throughout this main results section. 
The index \(t \in \{0,\dots,T-1\}\) denotes the within-episode time step in a finite-horizon trajectory, whereas \(Z \in \{1,\dots,N\}\) denotes the index of a recent policy snapshot (or sampling round) used in the temporal expressivity analysis. 
For the trainability results, it is also convenient to use the time-augmented state \(\bar S_t := (t,S_t)\). We first define the fixed-policy version of MI-TET, which is the quantity used in the trainability results. 
Since we work in a finite-horizon setting, we aggregate time steps by introducing an auxiliary index \(J\) over \(\{0,\dots,T-1\}\) with weights proportional to \(\gamma^t\).
\begin{definition} [Instantaneous MI-TET] 
    Let \(Y_t \in \{G_t,\; Q_t^{\pi_\theta}(S_t,A_t)\}\). 
Assume \(Y_{\min} \le Y_t \le Y_{\max}\), and fix a bin count \(B \in \mathbb{N}\). 
Define
\[
\Delta := \frac{Y_{\max}-Y_{\min}}{B} > 0,
\]
and partition \([Y_{\min},Y_{\max}]\) into intervals
\[
I_k := [b_k,b_k+\Delta), \qquad k=1,\dots,B-1,
\]
with \(b_1=Y_{\min}\), \(b_{k+1}=b_k+\Delta\), and \(I_B := [b_B,Y_{\max}]\). 
For each bin, define its midpoint
\[
m_k := b_k + \frac{\Delta}{2}, \qquad k=1,\dots,B.
\]
For each \(t\), let \(\tilde Y_t = k\) if \(Y_t \in I_k\), and define \(g(\tilde Y_t):=m_{\tilde Y_t}\).

Now define
\[
c_t := \frac{\gamma^t}{\sum_{j=0}^{T-1}\gamma^j}, \qquad t=0,\dots,T-1,
\]
and let \(J\) be an auxiliary random index with \(\Pr(J=t)=c_t\). 
Then define
\[
\bar S := (J,S_J), \qquad A := A_J, \qquad \tilde Y := \tilde Y_J.
\]
The instantaneous MI-TET is defined by
\[
\mathrm{MI\text{-}TET}_{\mathrm{inst}}(\theta) := I(A;\tilde Y \mid \bar S).
\]
\end{definition}
\noindent In short, \(\mathrm{MI\text{-}TET}_{\mathrm{inst}}(\theta)\) measures the conditional mutual information between the action and the discretized reward-related signal at the current policy parameter \(\theta\), after averaging across time steps through the auxiliary index \(J\). 
This is the version of MI-TET that appears in the trainability theorem below. 

First, one might ask why this discretization process is necessary. The primary motivation is to ensure a robust and computationally simple metric. In many reinforcement learning tasks, the action is often given to be discrete while the $Y$ is often continuous, instead; this is because the reward function $R$ is often taken to be continuous. If we were to use the continuous variable $Y$ directly, calculating the mutual information $I(A; Y | \bar S)$ would first require us to estimate the underlying continuous probability density (e.g. $p(y | a, \bar s)$). This step is not reachable by itself, rather forcing us to employ complex methods like Kernel Density Estimation (KDE) or other function approximation tools such as Neural Networks. Whichever method we use, this would add significant computational overhead for estimating $I(A; \Tilde{Y} | \bar S)$, thus making the online tracking of MI-TET unfeasible.

Moreover, even if $Y$ is given to be discrete, there exist some chances that $A$ and $Y$ might not \emph{match} that well; for instance, $Y$ is already divided into much more segments than $A$ is. This would make the given value of the mutual information become almost near to zero, since only sparse number of samples might be included into each bin of the histogram. 

Therefore, we employ a deliberate, a priori discretization of $Y$ into $\Tilde{Y}$. This process is mainly achieved by introducing the bin count $B$ as a crucial hyperparameter. If we set $B$ too large, we replicate the second problem mentioned above: the resulting histogram becomes too sparse, and the estimated mutual information $I(A; \Tilde{Y} | \bar S)$ may become vanishingly small. Conversely, if $B$ is set too small (e.g. $B = 1$), $\Tilde{Y}$ becomes a constant, and the mutual information trivially becomes zero, losing all the information. Yet, it is still \emph{better than nothing}; by properly selecting $B$ in advance, we aim to adjust the mutual information to a reasonably interpretable level. 

Second, the vast motivation of MI-TET itself. Intuitively, the ultimate goal of reinforcement learning can be simplified into the following one sentence: \emph{Let's increase the volume of what I expect to get by adjusting what I do.} This naturally motivates us to think \emph{What if we just quantify how much information the action distribution itself has about the reward distribution?} and as a mathematical way, one might easily choose mutual information for that purpose. Moreover, it is reinforcement learning-friendly, which almost every of the conventional metrics does not have as a property. There are two words, \textbf{exploration} and \textbf{exploitation} when we try to summarize the typical reinforcement learning process. Exploration prescribes the agent's behavior where it intentionally attempts a new action randomly to find a better policy, whereas exploitation signifies rather agents concentrating on selecting the most \emph{good looking} behavior learned to date when choosing the new action. In most reinforcement learning processes, the agent first tries exploration for most of the time, while gradually increasing the rate of exploitation with a decrease in the portion of exploration~\cite{sutton2018reinforcement}. In the reference~\cite{ReUploadingPQC}, we must carefully look at the role of $\beta$ in the \emph{softmax}-PQC, which is called an inverse-temperature parameter. This is the very parameter controlling \emph{greediness}, i.e. the dynamic of those two concepts, exploration and exploitation, and in the work the authors said they employed a linear annealing schedule starting from 1 and increasing up to the final $\beta$ (task-specific, prefixed value). The proposed metric, MI-TET, then can be understood with this $\beta$, showing the following relationships: 
\begin{enumerate}
    \item \textbf{During effective exploration} (low $\beta$, stochastic policy), MI-TET could be relatively high. \item \textbf{During exploitation} (high $\beta$, policy converges), the agent's policy becomes deterministic, i.e. $\pi(A | S) \rightarrow 1.0$ for the best action. As the policy's action entropy $H(A | \bar S)$ drops to zero, the mutual information $I(A; \Tilde{Y} | \bar S)$ must also drop to zero with trivial inequality $I(A; \Tilde{Y} | \bar S) \leq H(A | \bar S)$.
\end{enumerate} 
Thus by using the MI-TET, one can quantify and track how much the action distribution space the policy spews out narrows and gradually gets `concentrated' over time, which are both quite important in analyzing the learning performance. \\

\subsection{Theorems on Trainability and Expressivity}
What is interesting with this MI-TET is that we can upper bound some values that are related to trainability and temporal expressivity, respectively, with the use of MI-TET. First, let us introduce some theorems of trainability, which give us the upper bound for the norm of scaled gradients. To do so, one needs to understand what one-shot game is.

\begin{definition}[\emph{Informal}, One-shot game]
    One-shot game is basically the same as the case when the finite horizon $T$ is equal to one; that is, the state is fixed for $s$ and the agents picks the action $a \in A$ based on the policy $\pi_\theta (a | s) =: \pi_\theta(a)$. As a result, it gets a reward $R(s, a) =: R(a) \in \mathbb{R}$. Our job is to find an optimal parameter $\theta^*$ defined as
    \begin{equation}
        \theta^* = \argmax_\theta \mathbb{E}_{a \sim \pi_\theta}\left[R(a)\right].
    \end{equation}
    By playing this one-shot, immediate games multiple times, we aim to get closer to the optimal parameter $\theta^*$.
\end{definition}

Our proof for the trainability theorem begins with giving a proof for the simplified version, trainability theorem for just the one-shot game, and inducing that result further to much broader, generalized cases. We even discriminate the easy cases within two subcases, where the one with original $Y$ and discretized $\Tilde{Y}$.

%%% (1) 대신에 I로 바꾸기?

\begin{theorem}[Trainability Theorem (1) - One-shot game, non-discretized]\label{thm:1}
    Let our objective function $\eta_s(\theta) := \mathbb{E}\left[R(a)\right]$, and define the score function $S_\theta(a)$ as 
    \begin{equation}
        S_\theta(a) := \nabla_\theta \log \pi_\theta(a).
    \end{equation}
    Let $g_s(\theta)$ denote the gradient of our objective function $\eta_s$ and assume
    \begin{equation}
        ||S_\theta(a)|| \leq G_{\max}, \quad |R(a)| \leq R_{\max}.
    \end{equation}
    Then the following holds: 
    \begin{equation}
        ||g_s(\theta)|| \leq \sqrt{2}G_{\max}\sqrt{
        \operatorname{Var}(R)
        }\sqrt{I(A; R)}.
    \end{equation}
\end{theorem}
Note that some notations and MI-TET have changed due to the assumptions of the theorem. If we reach further to the discretized case, we get
\begin{theorem}[Trainability Theorem (2) - One-shot game, discretized]\label{thm:2}
    With the same notation given in the \Cref{thm:1}, 
    \begin{equation}
        ||g_s(\theta)|| \leq \sqrt{2}G_{\max}\sqrt{\operatorname{Var}(g(\Tilde{R}))}\sqrt{I(A; \Tilde{R})} + G_{\max}\Delta/2.
    \end{equation}
\end{theorem}
Again note that by discretizing $Y$, we encounter $G_{\max}\Delta/2$ as a cost. Yet, we take the $\Tilde{Y}$ for the real MI-TET, due to the the reasons given earlier. 

Next is our main trainability theorem. Using the instantaneous MI-TET defined above, together with the time-augmented state \(\bar S=(J,S_J)\), we obtain the following finite-horizon trainability bound.
\begin{figure*}[t]
    \begin{theorem}[Trainability Theorem (3) - Multiple games, discretized]
        Assume $|Y| \leq Y_{\max}$, and let
\[
\alpha_T(\gamma) := \left(\sum_{t=0}^{T-1}\gamma^t\right)^{-1},
\qquad
\eta'(\theta) := \alpha_T(\gamma)\,\eta(\theta).
\]
        Then
\[
\|\nabla_\theta \eta'(\theta)\|
\le
\sqrt{2}\,G_{\max}
\sqrt{\mathbb{E}_{\bar S}\!\left[\operatorname{Var}(g(\tilde Y)\mid \bar S)\right]}
\sqrt{I(A;\tilde Y\mid \bar S)}
+
G_{\max}\Delta/2,
\]
where \(Y_t \in \{G_t,\;Q_t^{\pi_\theta}(S_t,A_t)\}\). \label{thm:trainability}
    \end{theorem}
\end{figure*}
According to that, MI-TET has shown to be able to work as an \textbf{upper bound} for our modified (scaled) objective function's gradient; hence an indirect trainability proxy. Note that when we measure our policy function (or circuit)'s trainability or keep track of it, we use the gradient of the modified objective function $\nabla_\theta \eta'(\theta)$. On the other hand, for the real update process, we use the information of the original objective function's gradient $\nabla_\theta \eta(\theta)$.

At first glance, the gradient inequality may appear loose because it includes the maximum value of the score function's norm, $G_{\max}$, and discretization error term $\Delta$. 
However, note that our goal is not to achieve a numerically tight bound, but to provide an interpretable decomposition of the temporal variation in trainability during learning.
For convenience, define
\[
\sigma_{g\mid \bar S}
:=
\sqrt{\mathbb{E}_{\bar S}\!\left[\operatorname{Var}(g(\tilde Y)\mid \bar S)\right]}.
\]
Then the bound can be written as
\[
\|\nabla_\theta \eta'(\theta)\|
\le
a\,\sigma_{g\mid \bar S}\,\sqrt{\mathrm{MI\text{-}TET}_{\mathrm{inst}}(\theta)} + b,
\]
where
\[
a := \sqrt{2}\,G_{\max}, \qquad b := G_{\max}\Delta/2.
\]
In practical policy-gradient pipelines, $G_{\max}$ and $\Delta$ are fixed by design for a chosen architecture and binning scheme. 
Under this fixed-design setting, the principal time-varying factor is the product
\begin{equation}
    \sigma_{g | \bar S} \sqrt{\operatorname{MI-TET}_{\text{inst}}(\theta)}
\end{equation}
which serves as a tractable online proxy for trainability. 

For the complete proof of these trainability theorems, please refer to the \Cref{Appendix A}.

We now turn to temporal expressivity, while accordingly introducing a windowed counterpart of MI-TET that is estimated from pooled recent samples. We first propose the intuition behind the new expressivity definition: the policy function's temporal expressivity as the \textbf{deviation} among multiple sampled action distributions.

\begin{definition}[Windowed Temporal Expressivity]
    Let $(S, A, Z)$ denote the pooled joint distribution formed from the most recent $N$ policy snapshots, where $Z \in \{1, 2, \cdots, N\}$ indexes the snapshot. For a fixed state $s$, define
    \begin{equation}
        w_i(s) := P(Z = i | S = s), \quad \pi_i^s(a) := P(A = a | Z = i, S = s),
    \end{equation}
    and the state-conditional mean action distribution
    \begin{equation}
        \bar{\pi}^s(a) := \sum_{i=1}^N w_i(s) \pi_i^s(a) = P(A = a | S = s).
    \end{equation}
    We define the state-conditional windowed temporal expressivity by
    \begin{equation}
        \operatorname{Expr}(s) := \sum_{i=1}^N w_i(s) D_{KL}\left(\pi_i^s(\cdot) || \bar{\pi}^s(\cdot)\right).
    \end{equation}
    The overall windowed temporal expressivity is
    \begin{equation}
        \operatorname{Expr} := \mathbb{E}_S\left[\operatorname{Expr}(S)\right].
    \end{equation}
    Equivalently,
    \begin{equation}
        \operatorname{Expr} = I(A; Z | S).
    \end{equation}
\end{definition}
The final equality holds trivially due to definition of mutual information.

This temporal expressivity quantifies how much a policy's action distribution changes over time during training. Concretely, for a fixed state $s$, we consider the action distributions $\{\pi_i^s\}_{i=1}^N$ collected from the most recent $N$ sampling iterations (or policy snapshots). If these distributions differ substantially from one another, it indicates that the policy is still undergoing meaningful updates--e.g., it remains exploratory or continues to adapt strongly. In contrast, if the distributions become nearly identical, the policy has stabilized or is converging towards concentrated (nearly deterministic) behavior. In this sense, our \emph{Expr} captures not only whether the policy is stochastic or deterministic at a given moment, but more directly the temporal variability of the policy's behavior throughout learning.

We adopt the weighted Jensen-Shannon divergence (JSD) because it provides a symmetric and stable notion of deviation among multiple distributions. Importantly, the JSD-based definition admits a natural information-theoretic interpretation: it can be equivalently expressed as the conditional mutual information $I(A; Z | S)$. This means that, given a state, temporal expressivity measures how informative the action $A$ is about the temporal index $Z$--or equivalently, how much the policy's behavioral output changes across time.

This temporal expressivity metric is also meaningfully different from common alternatives. For instance, the policy entropy $H(A | S)$ describes how stochastic the policy is at a single time point, but it does not directly capture how the policy distribution evolves across training. A policy can maintain high entropy while remaining nearly unchanged over time, and conversely, entropy may decrease even though the policy undergoes abrupt changes during certain phases of training. Thus, entropy is informative about instantaneous randomness, but it is not designed to measure temporal evolution.

To connect this quantity with the windowed MI-TET, define
\[
\mathrm{MI\text{-}TET}_{\mathrm{win}} := I(A;\tilde Y \mid S),
\]
where the joint distribution is induced by the pooled recent-window samples. Then our expressivity theorem is given as follows. 

\begin{theorem}[Expressivity Theorem]
Let
\[
\mathrm{Expr}_{\mathrm{win}} := I(A;Z\mid S).
\]
Then
\[
\mathrm{Expr}_{\mathrm{win}}
\le
\mathrm{MI\text{-}TET}_{\mathrm{win}}
+
I(A;Z\mid \tilde Y,S).
\]
If $I(A; Z | \tilde{Y}, S) \approx 0$, then
\begin{equation}
    \mathrm{Expr}_{\mathrm{win}} \leq \mathrm{MI\text{-}TET}_{\mathrm{win}}.
\end{equation}
\end{theorem}
\begin{proof}
    \begin{align}
        I(A; \Tilde{Y}, Z | S) & = I(A; \Tilde{Y} | S) + I(A; Z | \Tilde{Y}, S) \\
        & = I(A; Z | S) + I(A; \Tilde{Y} | Z, S)
    \end{align}
    Then
    \begin{equation}
        I(A; Z | S) = I(A; \Tilde{Y}, Z | S) - I(A; \Tilde{Y} | Z, S).
    \end{equation}
    Since $I(A; \Tilde{Y} | Z, S) \geq 0$, we get
    \begin{align}
        & I(A; \Tilde{Y}, Z | S) \geq I(A; Z | S) \\
        & \equiv I(A; \Tilde{Y} | S) + I(A; Z | \Tilde{Y}, S) \geq I(A; Z | S).
    \end{align}
\end{proof}

This expressivity theorem decomposes temporal expressivity into two components: the part mediated by the discretized reward signal $\Tilde{Y}$, and a residual component that remains time-dependent even after conditioning on $(\Tilde{Y}, S)$. The assumption $I(A; Z | \Tilde{Y}, S) \simeq 0$ means that, given the state and (discretized) reward, the policy's action distribution does not depend on the temporal index $Z$; i.e., the policy is \emph{locally} stationary under comparable state-reward conditions. Then the theorem says that under that regime, MI-TET would serve as a valid upper bound for our temporal expressivity metric with constant factor one. 

\subsection{Initialization-time Probabilistic Prescreening}
While the above result links MI-TET to temporal expressivity, it is also useful to note that MI-TET can be connected back to initialization-time trainability in a probabilistic manner. By combining the trainability bound with a concentration-type assumption over the initialization distribution, one can derive a one-sided upper bound for the probability that the initialized policy retains a gradient norm above a prescribed threshold. Hence, beyond tracking temporal variability during learning, MI-TET can also induce an assumption-based prescreening criterion that can help rule out architectures whose random initializations are likely to be gradient-fragile.

First, let us introduce some useful notations.
\begin{itemize}
    \item \(\theta \sim D_{\mathrm{init}}\) (initial parameter distribution),
    \item \(G(\theta) := \|\nabla_\theta \eta'(\theta)\|\),
    \item \(M(\theta) := \mathrm{MI\text{-}TET}_{\mathrm{inst}}(\theta) = I(A;\tilde Y\mid \bar S)\),
    \item \(V(\theta) := \mathbb{E}_{\bar S}\!\left[\operatorname{Var}(g(\tilde Y)\mid \bar S)\right]\),
    \item \(a := \sqrt{2}\,G_{\max}\), \(b := G_{\max}\Delta/2\),
    \item \(E_\epsilon := \{G(\theta)\ge \epsilon\}\).
\end{itemize}
Then our \Cref{thm:trainability} simplifies into
\begin{equation}
    G(\theta) \leq a\sqrt{V(\theta)}\sqrt{M(\theta)} + b.
\end{equation}

To turn this pointwise inequality into a probabilistic statement over the initialization distribution, we impose the following two regularity assumptions.

\begin{assumption}[Initialization-time variance envelope]
There exists a deterministic constant $\bar V > 0$ such that
\[
V(\theta) \le \bar V
\qquad
\text{for } D_{\mathrm{init}}\text{-almost every } \theta.
\]
\end{assumption}

\begin{assumption}[Initialization concentration of MI-TET]
Let
\[
\mu_M := \mathbb{E}_{\theta \sim D_{\mathrm{init}}}[M(\theta)].
\]
Assume that there exist constants $C_0 \ge 1$ and $\kappa > 0$ such that, for every $t \ge 0$,
\[
\Pr_{\theta \sim D_{\mathrm{init}}}\!\big(M(\theta)-\mu_M \ge t\big)
\le
C_0 e^{-\kappa t^2}.
\]
\end{assumption}

Under these assumptions, we obtain the following initialization-time prescreening result.

\begin{proposition}[Initialization-time probabilistic prescreening under concentration assumptions]
Assume that the above trainability bound holds, together with Assumptions 1 and 2. Then, for every $\epsilon > b$,
\[
\Pr_{\theta \sim D_{\mathrm{init}}}\!\big(E_\epsilon\big)
=
\Pr_{\theta \sim D_{\mathrm{init}}}\!\big(G(\theta)\ge \epsilon\big)
\le
C_0 \exp\!\left(
-\kappa\big(\tau_\epsilon-\mu_M\big)_+^2
\right),
\]
where
\[
\tau_\epsilon
:=
\left(
\frac{\epsilon-b}{a\sqrt{\bar V}}
\right)^2,
\qquad
(x)_+ := \max\{x,0\}.
\]
\end{proposition}

\begin{proof}
By Assumption 1, we have
\[
G(\theta)
\le
a\sqrt{V(\theta)}\sqrt{M(\theta)} + b
\le
a\sqrt{\bar V}\sqrt{M(\theta)} + b.
\]
Hence, for any $\epsilon>b$,
\[
G(\theta)\ge \epsilon
\quad\Longrightarrow\quad
a\sqrt{\bar V}\sqrt{M(\theta)} + b \ge \epsilon,
\]
which implies
\[
M(\theta)\ge
\left(
\frac{\epsilon-b}{a\sqrt{\bar V}}
\right)^2
=
\tau_\epsilon.
\]
Therefore,
\[
E_\epsilon
\subseteq
\{M(\theta)\ge \tau_\epsilon\},
\]
and thus
\[
\Pr(E_\epsilon)
\le
\Pr\big(M(\theta)\ge \tau_\epsilon\big).
\]

If $\tau_\epsilon \le \mu_M$, then the claim is trivial since
\[
\Pr(E_\epsilon)\le 1 \le C_0.
\]
If $\tau_\epsilon > \mu_M$, let $t=\tau_\epsilon-\mu_M \ge 0$. Then Assumption 2 gives
\[
\Pr\big(M(\theta)\ge \tau_\epsilon\big)
=
\Pr\big(M(\theta)-\mu_M \ge \tau_\epsilon-\mu_M\big)
\le
C_0 e^{-\kappa(\tau_\epsilon-\mu_M)^2}.
\]
Combining both cases yields
\[
\Pr(E_\epsilon)
\le
C_0 \exp\!\left(
-\kappa\big(\tau_\epsilon-\mu_M\big)_+^2
\right).
\]
\end{proof}
The above proposition should be interpreted as a one-sided elimination rule rather than a calibrated predictor of optimization success. In particular, the bound becomes informative only when $\tau_\epsilon > \mu_M$. In that regime, the quantity $\tau_\epsilon - \mu_M$ directly controls the sharpness of the exponential probability decay, so that a larger gap implies a stronger suppression of the probability that a random initialization retains a non-negligible gradient norm. Conversely, when $\tau_\epsilon \le \mu_M$, the resulting upper bound becomes vacuous, indicating that the present criterion does not provide meaningful elimination power in that regime.

Motivated by this observation, we introduce the following bound-inspired prescreening score.
\begin{definition}[Initialization-time prescreening score]
For a fixed task, initialization protocol, and gradient threshold $\epsilon > b$, define the initialization-time prescreening score by
\[
\Gamma_\epsilon
:=
\big(\tau_\epsilon - \mu_M\big)_+,
\]
where
\[
\tau_\epsilon
=
\left(
\frac{\epsilon-b}{a\sqrt{\bar V}}
\right)^2,
\qquad
\mu_M
=
\mathbb{E}_{\theta\sim D_{\mathrm{init}}}[M(\theta)].
\]
\end{definition}
In this sense, $\Gamma_\epsilon$ could be used as an assumption-based criterion for prescreening architectures that are likely to be initialization-fragile.

In practice, the exact quantity $\Gamma_\epsilon$ may not be directly accessible, since $\mu_M$ may need to be approximated empirically and the theoretical constants such as $G_{\max}$ might be difficult to estimate. Accordingly, under a common task and fixed initialization procedure, one may instead use the corresponding plug-in estimate
\[
\widehat{\Gamma}_\epsilon
:=
\big(\widehat{\tau}_\epsilon - \widehat{\mu}_M\big)_+,
\]
where $\widehat{\mu}_M$ is an empirical estimate of the initialization-time mean of MI-TET, and $\widehat{\tau}_\epsilon$ is obtained by replacing the theoretical quantities in $\tau_\epsilon$ with their protocol-calibrated or sample-based counterparts (for more details, please refer to \Cref{Appendix B}). Importantly, $-\widehat{\Gamma}_\epsilon$ should not be interpreted as a calibrated estimate of the survival probability itself; rather, it serves as a bound-inspired ranking heuristic for comparing architectures under a common evaluation setting.

\section{Numerical Simulations}
\subsection{Experimental Protocol}
This numerical simulation study is designed to examine MI-TET as a temporal diagnostic within a quantum policy gradient (REINFORCE) pipeline and to assess whether the observed training dynamics are qualitatively consistent with the inequality-based relationships developed in the preceding sections. In particular, we examine how MI-TET evolves during training and how it co-varies with gradient-related and temporal expressivity-related terms appearing in the theoretical bounds.

All experiments use the CartPole-v1 environment in a finite horizon reinforcement learning setting \cite{OpenAIGym}, consistent with the assumptions adopted in the policy gradient formulation. We employ a  REINFORCE pipeline with the softmax-PQC policy architecture of Jerbi et al. \cite{ReUploadingPQC}. For the primary training runs, we use a 4-qubit, depth-1 circuit with all-to-all trainable entanglers. Training is performed for 180 episodes over six random seeds, using mini-batches of eight episodes, a horizon cap of 180 steps per episode, and $\gamma = 1$. Greedy policy evaluation is performed every four parameter updates, as well as at the first update over six evaluation episodes. The inverse-temperature parameter $\beta$ is linearly annealed from 1.0 to 3.0 during training, and no value baseline is used.

To monitor the key, theory-linked quantities during training, we estimate MI-TET, the trainability proxy (scaled gradient and right-hand side term of the trainability inequality), the temporal expressivity metric, and the residual term online at each parameter update using rolling windows of recently collected trajectory samples. In the main runs, these diagnostics are recomputed from a sliding window of 1200 samples, and estimation starts once at least 140 samples are available. States are discretized using four bins per state dimension. Unless otherwise stated, we record both $G_t$- and $\widehat{Q}$ (empirical $Q$ function)-based reward signals. For both of them, we use $B = 16$ as the default discretization bin count.

As supplementary protocol checks, we also conduct a bin-sensitivity experiment on the same primary architecture with $B \in \{4, 8, 12, 16\}$, using four seeds and 180 training episodes under the same main-run schedule (batch size 8, $\beta$ : $1.0 \rightarrow 3.0$, no baseline). We also perform an architecture-comparison study with shallow4, mid4, deep4 PQCs with configurations (4, 1, ring, trainable), (4, 1, all-to-all, trainable), and (4, 2, all-to-all, trainable), where the entries denote the number of qubits, circuit depth, entanglement topology, and entangler type, respectively. For the initialization-time prescreening experiment, we sample 12 random initializations per architecture-seed pair and collect six episodes per initialization, using $\epsilon = 0.2$, strict state encoding, and both $G_t$- and $\widehat{Q}$-based reward signals; we use $B = 64$, a minimum of 80 samples, and an empirical quantile $|Y|$ calibration ($q = 0.85$). The follow-up architecture training is then run over five seeds, with 80 episodes for shallow4/deep4 and 160 episodes for mid4, using mini-batches of four episodes and $\beta$ annealing from 1.0 to 8.0.

For more details concerning reproducibility, please refer to the following GitHub repositories (\cref{githubLink}).

\subsection{Learning Dynamics in the Reference Setting}
We begin by summarizing the learning dynamics in the reference (primary-architecture) setting. We track both episode returns during training and evaluation, as well as MI-TET computed using $G_t$ and $\widehat{Q}$ as reward signals.

\begin{figure*}
    \centering
    \begin{subfigure}[t]{0.49\textwidth}
        \centering
        \includegraphics[width=\linewidth]{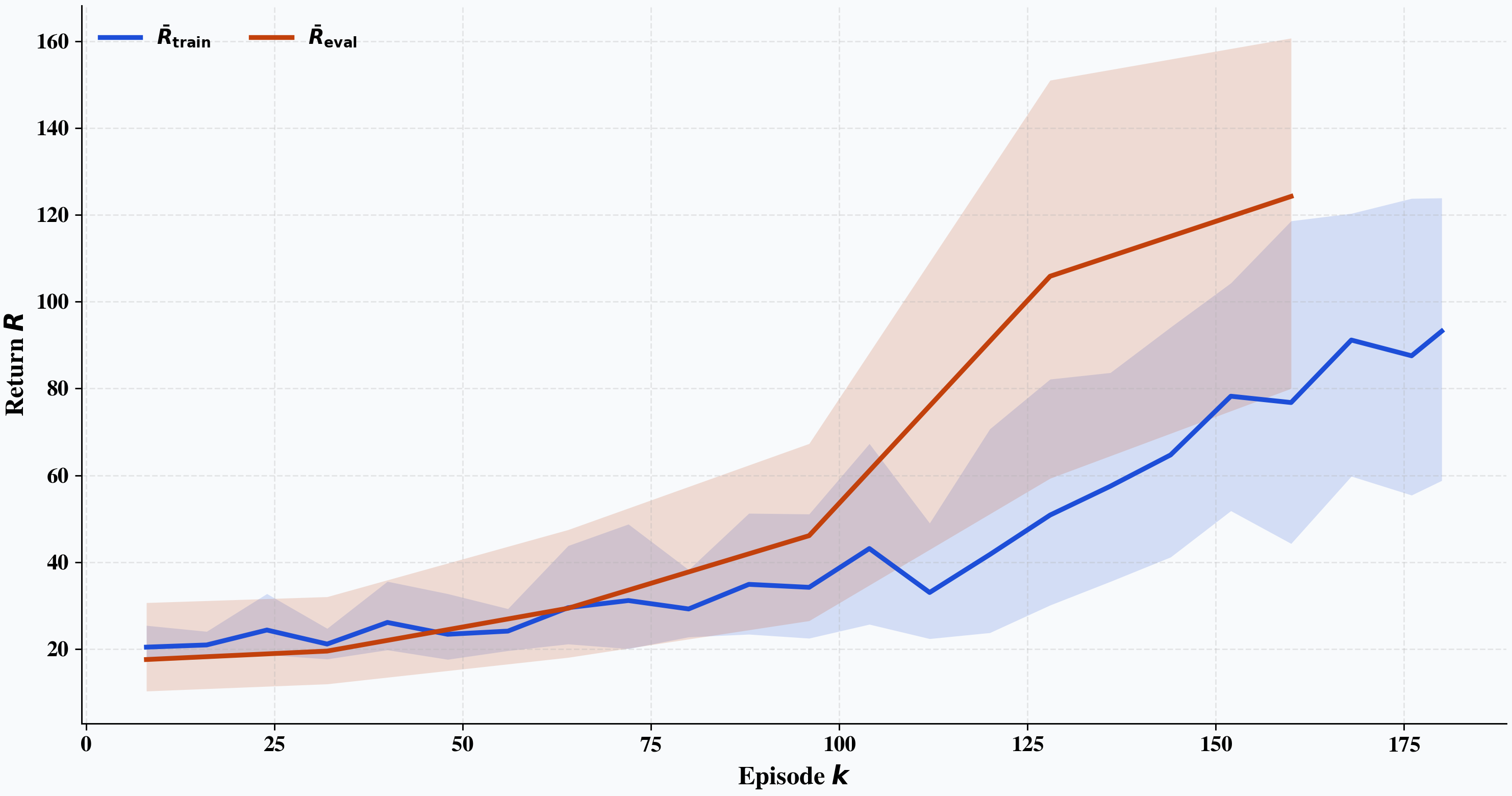}
        \caption{Training and evaluation episode return dynamics.}
        \label{fig:ref_return}
    \end{subfigure}
    \hfill
    \begin{subfigure}[t]{0.49\textwidth}
        \centering
        \includegraphics[width=\linewidth]{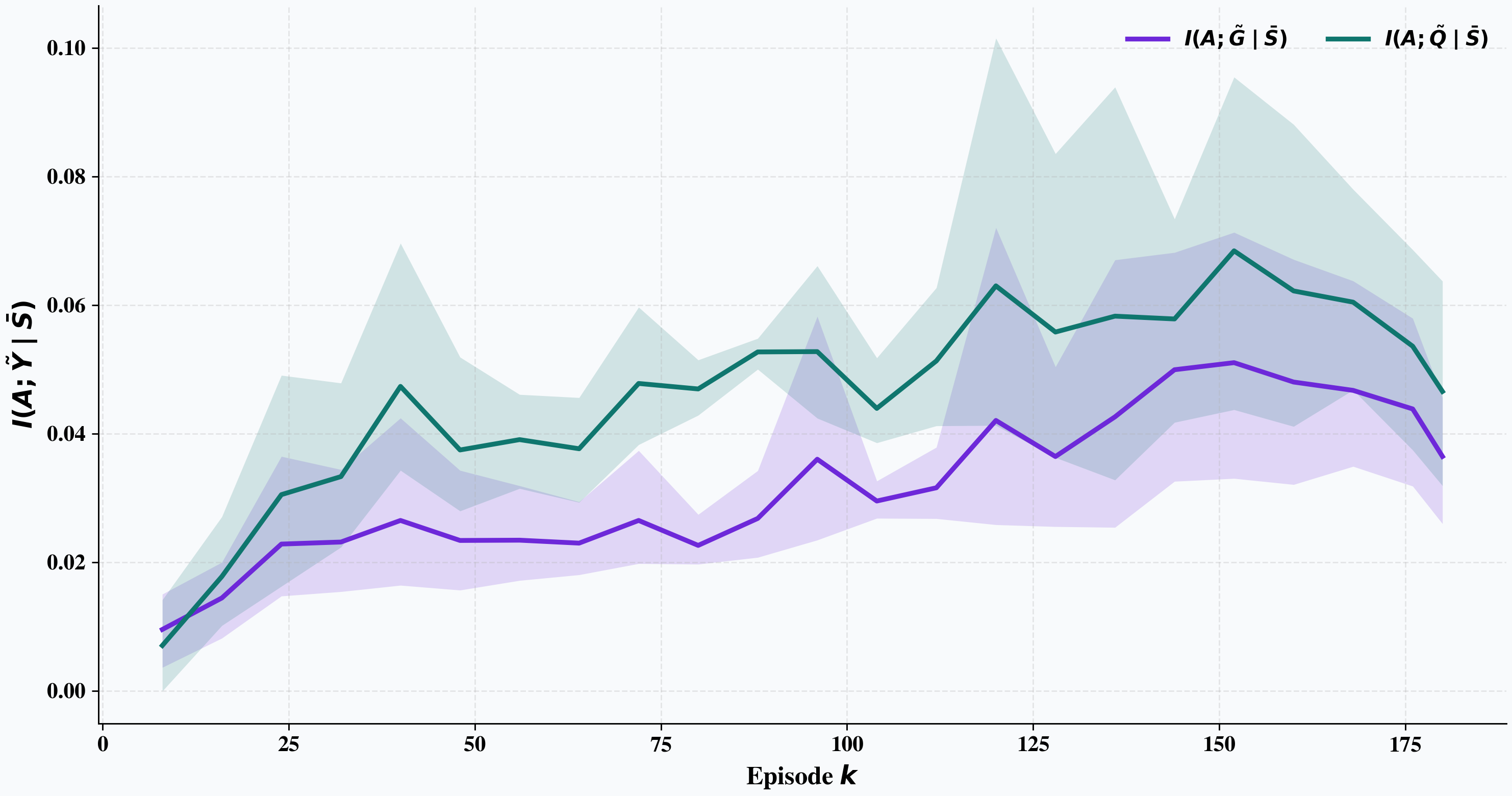}
        \caption{MI-TET trajectories with $G_t$- and $\widehat{Q}$-based signals.}
        \label{fig:ref_mitet}
    \end{subfigure}
    \caption{\textbf{Reference learning dynamics and online MI-TET tracking in the primary-PQC CartPole-v1 setting}. Shaded regions indicate variability across seeds.}
    \label{fig:ref_dynamics}
\end{figure*}

In the reference setting, the return curves indicate that the quantum policy gradient pipeline learns consistently across random seeds and reaches a sustained high-return regime. As shown in the first plot, both training and evaluation returns improve substantially over the course of learning. In the late stage, however, the rate of improvement appears to diminish, suggesting that performance may be approaching a plateau. Moreover, the evaluation curve is noticeably smoother, likely because evaluation is recorded every four updates and averaged over six episodes. By contrast, the training curve is obtained from on-policy interaction during learning and therefore exhibits greater variability.

In the second plot, the MI-TET trajectories provide a complementary view of how the action-reward dependence evolves during training. It shows that MI-TET increases steadily during the early stage of training, but its trend reverses once learning enter a partial plateau, decreasing between approximately episodes 150 and 175. This pattern is itself consistent with the exploration-exploitation dynamics anticipated in the theory section. During the early phase of learning, when the agent explores a diverse range of behaviors and searches for dependencies between action and returns, the mutual information--which may be interpreted as a quantity that captures this dependency structure--tends to increase. As training becomes more stable, however, the policy entropy associated with a given state decreases, which in turn leads to a reduction in MI-TET. From this perspective, MI-TET appears to track the expected learning dynamics reasonably well. In addition, the absence of a substantial `trend' discrepancy between MI-TET computed using $\widehat{Q}$ and that computed using $G_t$ is also consistent with the preliminary claims suggesting that those two can be treated interchangeably.

\subsection{Empirical Validation of Trainability Theorem}
Next, to empirically examine the trainability theorem, we plot the left-hand side scaled gradient norm together with the individual terms on the right-hand side in a single figure. 

\begin{figure*}
    \centering
    \includegraphics[width=\textwidth]{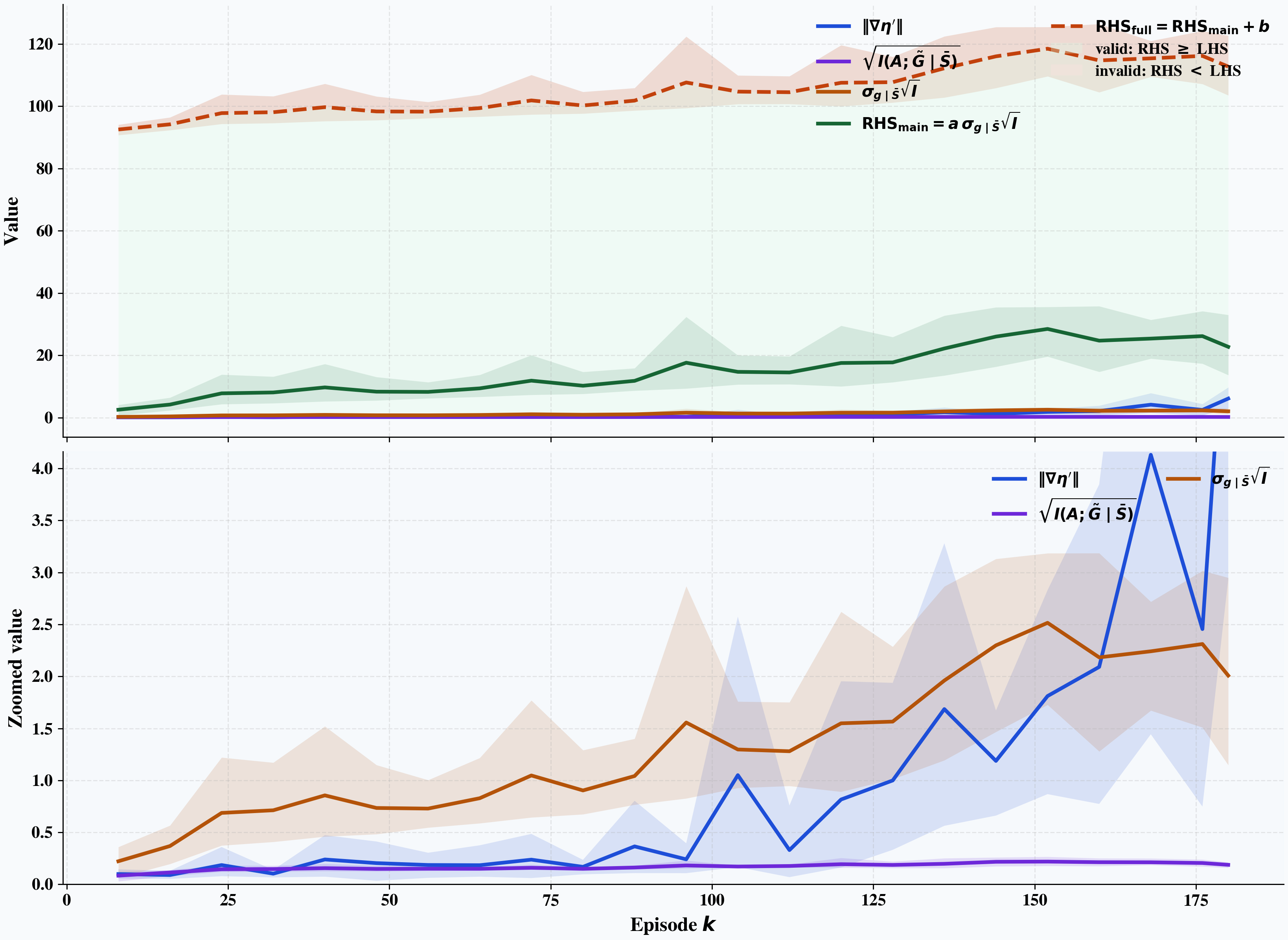}
    \caption{\textbf{Empirical verification of the trainability theorem in the reference setting}. 
    The top panel shows the scaled gradient norm and the RHS components on a common axis, while the bottom panel provides a zoomed-in view of the low-scale region.}
    \label{fig:trainability_theorem_fullpage}
\end{figure*}

The top panel of Fig.~2 displays $||\nabla_\theta \eta'(\theta)||$ (the scaled gradient norm), $\sqrt{I(A; \tilde{G} | \bar S)}$, $\sigma_{g | \bar S} \sqrt{I(A; \tilde{G} | \bar S)}$, as well as the full right-hand side term of the inequality on a common axis, whereas the bottom panel provides a zoomed-in view of the low-scale region around the gradient and multiplicative term. Here,
\begin{align}
    & \text{RHS}_{\text{main}} = a \sigma_{g | \bar S} \sqrt{I(A; \tilde{G} | \bar S)} \\
    & \text{RHS}_{\text{full}} = \text{RHS}_{\text{main}} + b
\end{align}
with constants
\begin{equation}
    a = \sqrt{2}G_{\max}, \quad b = \frac{G_{\max} \Delta}{2},
\end{equation}
where $\Delta$ corresponds to the discretization bias term.

The top panel shows that the upper bound `vacuumness' with the theoretically predicted scale mismatch is clearly present in practice as well. In particular, $||\nabla_\theta \eta'(\theta)||$ remains at roughly the $10^0$ scale, whereas $\text{RHS}_{\text{full}}$ reaches the $10^2$ scale, so the visual gap between the two is on the order of $10^2$. This is directly explained by the definition of the scaled objective function, 
\begin{equation}
    \eta'(\theta) = \frac{1 - \gamma}{1 - \gamma^T} \eta(\theta), \quad \nabla_\theta \eta'(\theta) = \left(\frac{1 - \gamma}{1 - \gamma^T}\right) \nabla_\theta \eta(\theta).
\end{equation}
When $\gamma = 1$ with a finite horizon $T$, as in the present experiment, the prefactor satisfies \begin{equation}
    \frac{1 - \gamma}{1 - \gamma^T} \rightarrow \frac{1}{T}
\end{equation}
so that for $T = 180$ it becomes approximately $5.6 \times 10^{-3}$. Thus, the absolute scale of the gradient is intrinsically reduced by approximately $10^2$-fold, which makes the observed scale gap essentially unavoidable. Moreover, if we decompose the RHS into $\text{RHS}_{\text{main}}$ and $b = G_{\max}\Delta / 2$, we see that $\text{RHS}_{\text{main}}$ alone--although still an upper bound--lies substantially closer to $||\nabla_\theta \eta'(\theta)||$ than $\text{RHS}_{\text{full}}$. By contrast, once the second term $b$ is added, the bound becomes much larger, indicating that this term is the primary source of the bound's triviality.

The bottom (zoomed-in) panel makes it clearer that the time-varying core factor
\begin{equation}
    \sigma_{g | \bar S} \sqrt{I(A; \tilde{G} | \bar S)}
\end{equation}
does track the variant evolution of the actual gradient to a meaningful extent. More specifically, the dominant time-varying factor, $\sigma_{g | \bar S} \sqrt{I(A; \tilde{G} | \bar S)}$, was observed to exhibit notable correlation with the scaled gradient norm during the early and middle stages of training, while also providing an effective upper bound on it in practice. Given that both a and b are generally positive, this is theoretically encouraging, as it suggests that the inequality may already hold with the dominant term alone. The strong correlation in the early and middle phases is also supported quantitatively. Following the stage partition used in the corresponding GitHub repository, updates up to the 8th were classified as early, the 9th through 16th updates as middle, and updates from the 17th upward as late. Under this partition, the correlations between the scaled gradient norm and the dominant time-varying factor $\sigma_{g | \bar S} \sqrt{I(A; \tilde{G} | \bar S)}$ were analyzed as follows (for more completed information concerning correlation, please refer to full correlation matrix appearing in \cref{Appendix C}).
\begin{itemize}
    \item In the early stage, Pearson = 0.7497 and Spearman = 0.1879.
    \item In the middle stage, Pearson = 0.6619 and Spearman = 0.7872.
    \item In the late stage, Pearson = 0.0464 and Spearman = 0.3086.
    \item Over the full training history, Pearson = 0.4298 and Spearman = 0.7271.
\end{itemize}
The relatively strong correlations in the early and middle stages are particularly meaningful, since these phases generally correspond to the period in which learning is most active, thus making its current trainability estimation important. From this perspective, the result strengthens the claim for MI-TET as a diagnostic indicator of learnability. Although the correlation weakens in the late stage, this seems natural given the characteristic behavior of policy gradient learning algorithms, in which gradient estimation noise counts substantially throughout the whole training. Indeed, if the late stage is interpreted as a regime of gradient fluctuations around a certain plateau level (about 3.25), the dominant factor term $\sigma_{g | \bar S} \sqrt{I(A; \tilde{G} | \bar S)}$ may still be understood as tracking the overall tendency of the gradient norm, functioning in a manner analogous to somewhat a moving average.

\subsection{Empirical Validation of Expressivity Theorem}
We also examine the time-varying behavior of the temporal expressivity quantity by comparing it with the policy entropy, and we use the corresponding plot to assess whether the expressivity theorem is empirically satisfied--namely, whether MI-TET, when combined with the residual term as predicted by theory, serves as an appropriate upper bound.

\begin{figure*}
    \centering
    \includegraphics[width=\textwidth]{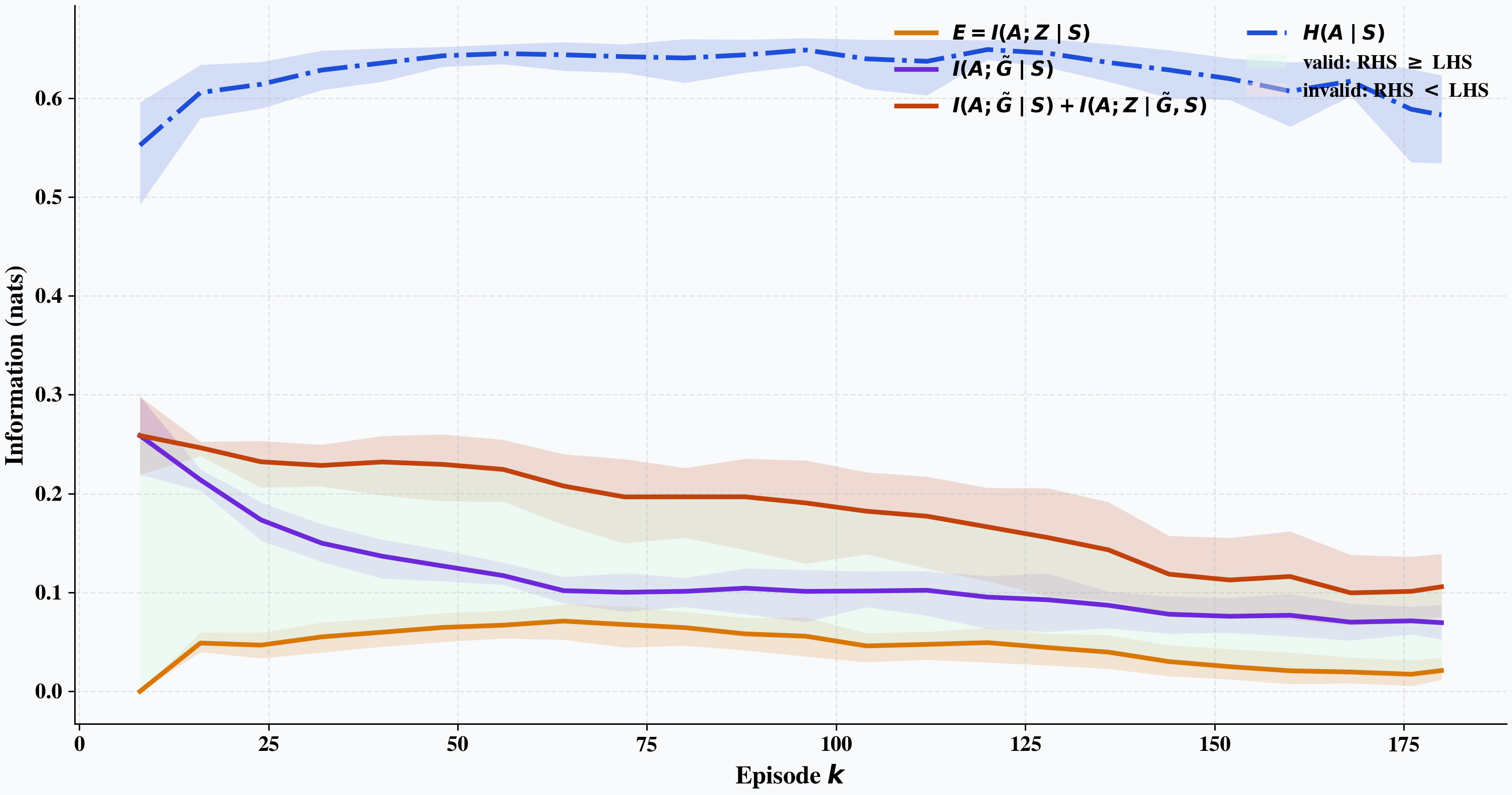}
    \caption{\textbf{Empirical verification of the expressivity theorem in the reference setting}.}
    \label{fig:trainability_theorem_fullpage}
\end{figure*}

In the main text, we emphasized that policy entropy and temporal expressivity need not evolve identically, and that this potential decoupling provides part of the motivation for introducing temporal expressivity as a distinct interpretive tool. In the present experiment, however, we do not observe a pronounced divergence between the two. Rather, the trajectories follow the familiar narrative of successful learning: as training proceeds, the agent's behavior becomes progressively more fixed, and accordingly both the policy entropy and the temporal expressivity--each reflecting, in different ways, the degree of behavioral variation--decrease gradually over time. At least at a qualitative level, the two quantities therefore display a broadly consistent downward trend in this reference setting.

In addition, the plot shows that the expressivity theorem is cleanly satisfied at all measurement points. Although the inclusion of the residual term renders the bound somewhat less tight, it still appears substantially tighter than the bound obtained from the trainability theorems. This is also consistent with what one would expect from the proof of the expressivity theorem itself. The theorem relies on comparatively strong information-theoretic facts, such as the chain rule and nonnegativity of mutual information, and therefore takes a form that is, in a sense, theoretically more complete and structurally better behaved. Of course, a separate analysis is still required for the locally stationary assumption underlying the treatment of the residual term; this will be carried out in the next section. Nevertheless, under the assumption that these conditions hold, the empirical result observed here is notably clean and well aligned with the theoretical prediction.

\subsection{Observations Concerning the Locally-Stationary Assumption}
The following two plots examine when the dynamics satisfy the locally stationary condition during learning. In the left panel, we track the magnitude of the residual term and regard the locally stationary condition as satisfied whenever it falls below a prescribed threshold ($\tau_{LS} = 0.05$). The right panel provides a complementary view by showing the probability of entering this locally stationary zone over the timestamp of learning.

\begin{figure*}
    \centering
    \includegraphics[width=\textwidth]{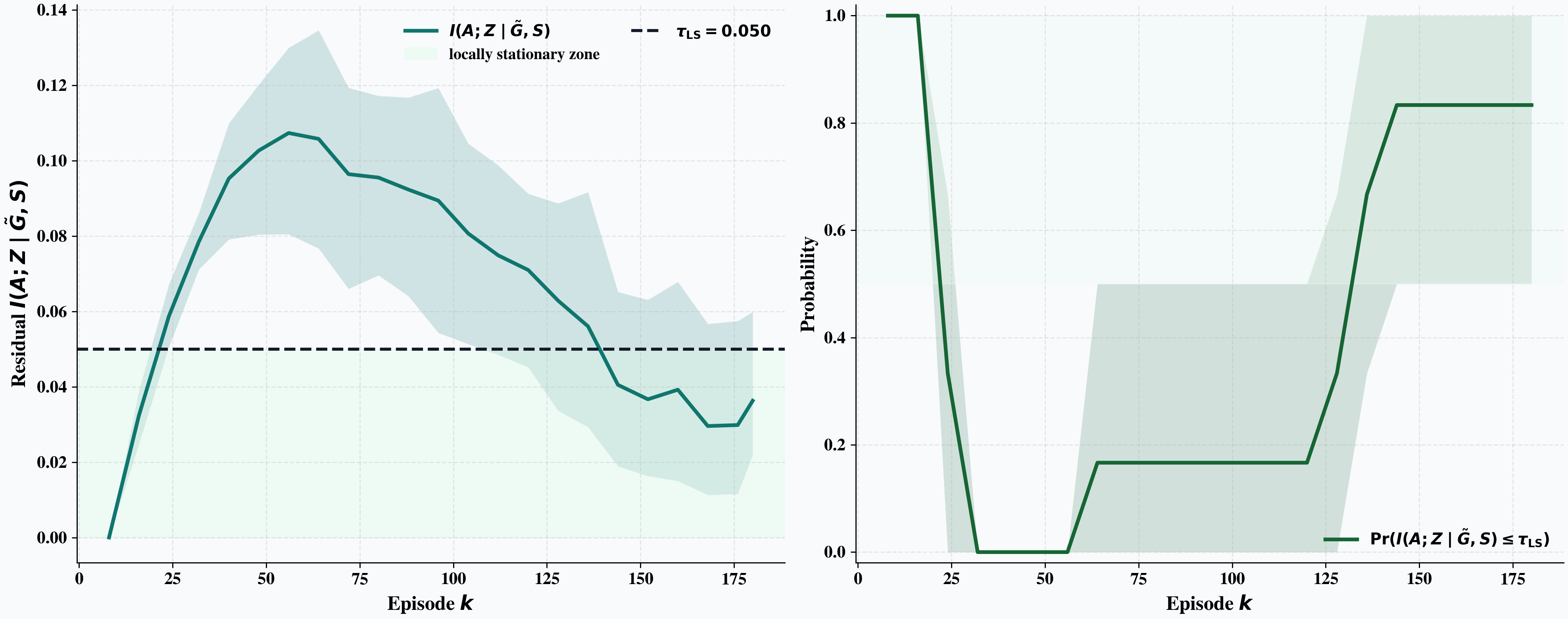}
    \caption{\textbf{Empirical assessment of the residual term and the local-stationarity regime in the reference setting}. \textbf{Left}: time evolution of the residual $I(A; Z | \tilde{G}, S)$ relative to the local-stationarity threshold $\tau_{LS} = 0.05$. \textbf{Right}: empirical probability $\operatorname{Pr}\left(I(A; Z | \tilde{Z}, S) \leq \tau_{LS}\right)$, indicating how often the dynamics lie within the locally stationary zone during learning.}
    \label{fig:trainability_theorem_fullpage}
\end{figure*}

The left panel indicates that the desired condition is not well satisfied during the early and intermediate stages of learning, but only approaches the threshold or is maintained near it in the very earliest and later stage. This is also consistent with the underlying intuition. The locally stationary condition,
\begin{equation}
    I(A; Z | \tilde{G}, S) \approx 0,
\end{equation}
means, in information-theoretic terms, that once the state and discretized return are given, the sampling time $Z$ provides almost no additional information about the action. Equivalently, the action depends primarily on the state-reward pair and only weakly on time, indicating that the policy is no longer changing substantially over time. During the early and middle stages of learning, however, learning is still actively reshaping the relationship among state, reward, and action, so it is more natural that this condition is not satisfied. Accordingly, the empirical results suggest that local stationarity should not be interpreted as a global property holding uniformly throughout learning, but rather as an approximate and gradual property that emerges mostly after learning has sufficiently progressed.

The right panel provides a more refined perspective on the same condition. The probability of satisfying the threshold condition is low--or in some intervals close to zero--during the early phase, and then gradually increases as training proceeds. However, even in the later stage, this probability typically remains with high volatility rather than uniformly staying close to one. This indicates that locally stationary behavior becomes more frequent in the later phase of learning, but substantial fluctuations may still remain even after the main learning dynamics have stabilized. In other words, the locally stationary regime becomes more plausible late in training, but it does not become uniformly dominant.

One particularly noteworthy observation is that, from around episode 150 upward--when the locally stationary condition begins to be satisfied more clearly--the gap between the full term (red curve) and the MI-TET term (blue curve) becomes noticeably narrower. This, too, may be regarded as experimental evidence that more firmly supports the theoretical discussion presented earlier. Conversely, the fact that the gap between the two curves is largest during the middle stage of learning (approximately episodes 25 to 125), where the locally stationary condition appears to be least well satisfied, further corroborates this interpretation.

Taken together with the preceding expressivity theorem results, this analysis leads to the following conclusion:
\begin{enumerate}
    \item The full form of the expressivity theorem remains stably valid throughout training.
    \item The residual term is not negligible in practice and is especially essential during the early and intermediate stages of learning.
    \item The local-stationary assumption does not hold automatically over the entire training horizon. Instead, the residual decreases as learning stabilizes, making the assumption progressively more plausible as an approximate description in the later stage.
\end{enumerate}

\subsection{Initialization-Time Prescreening Across Architectures}
We next examine the initialization-time prescreening score $\Gamma_\epsilon$ and assess how informative it is for downstream training behavior. Fig.~5 is organized accordingly: the three panels compare $\Gamma_\epsilon$ against the initialization survival rate $p_{\text{survive}}^{\text{init}}$ (left), the early failure rate $p_{\text{fail}}^{\text{early}}$ (middle), and the final-stage stability index $S_{\text{final}}$ (right).

\begin{figure*}[t]
    \centering
    \includegraphics[width=\textwidth]{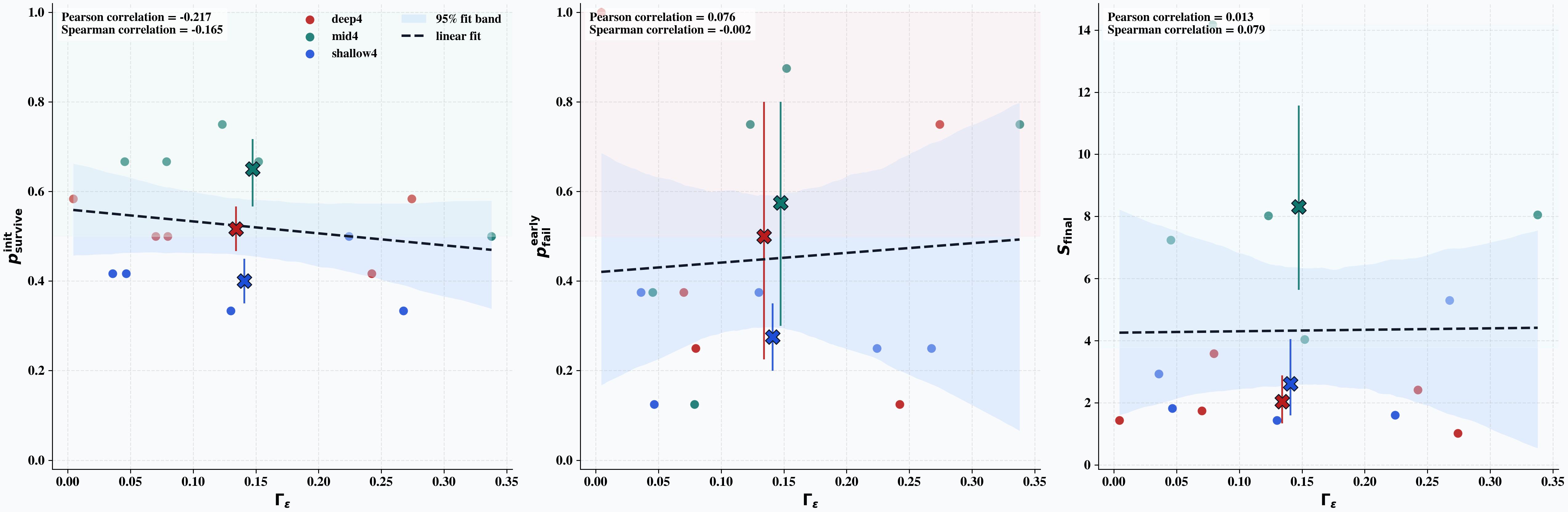}
    \caption{
    \textbf{Initialization-time prescreening protocol and downstream associations of \(\Gamma_{\epsilon}\).}
    Each panel relates the prescreening score \(\Gamma_{\epsilon}\) (computed at initialization) to a downstream summary statistic.
    \textbf{Left:} initialization survival rate
    \(p_{\mathrm{survive}}^{\mathrm{init}}
    := \frac{1}{N_{\mathrm{init}}}\sum_{j=1}^{N_{\mathrm{init}}}\mathbf{1}\!\left[\left\|\nabla_{\theta}\eta'(\theta_j)\right\|\ge t_{\mathrm{init}}\right]\),
    i.e., the fraction of random initializations whose scaled gradient norm exceeds the survival threshold \(t_{\mathrm{init}}\).
    \textbf{Middle:} early failure rate
    \(p_{\mathrm{fail}}^{\mathrm{early}}
    := \frac{1}{U_{\mathrm{early}}}\sum_{u=1}^{U_{\mathrm{early}}}\mathbf{1}\!\left[\left\|\nabla_{\theta}\eta'_u\right\|< t_{\mathrm{fail}}\right]\),
    i.e., the fraction of the first \(U_{\mathrm{early}}\) parameter updates whose scaled gradient norm falls below the failure threshold \(t_{\mathrm{fail}}\).
    \textbf{Right:} final-stage stability index
    \(S_{\mathrm{final}} := \bar R_{\mathrm{tail}}/(1+\sigma_{\mathrm{tail}})\),
    where \(\bar R_{\mathrm{tail}}\) and \(\sigma_{\mathrm{tail}}\) are the mean and standard deviation of the last \(K\) greedy-evaluation returns, respectively.
    Dots denote seed-level runs; \(\times\) markers indicate architecture-level means with bootstrap confidence intervals (vertical).
    The dashed line and shaded band show a linear fit and its 95\% fit band. In our analysis, \(U_{\mathrm{early}}=8\), \(t_{\mathrm{init}}=t_{\mathrm{fail}}=0.20\), and \(K=5\).
    }
    \label{fig:prescreening_protocol}
\end{figure*}

The left panel provides the most direct validation of the intended role of the score. Larger values of $\Gamma_\epsilon$ are associated with lower initialization survival rates, indicating that a larger prescreening score indeed corresponds to a more fragile starting point. This precisely matches the qualitative direction expected from the construction of the score: it is designed to flag architectures that are less likely to maintain a non-negligible gradient signal at the beginning. 

The middle panel indicates that the initialization-level fragility signal carries weak-to-modest information into the earliest stage of training, but with substantial noise. In particular, larger values of $\Gamma_\epsilon$ tend to be associated with higher early failure rates, which is directionally consistent with the idea that fragile starts are more likely to encounter vanishing- or unstable-gradient behavior shortly after optimization begins. At the same time, the dispersion is considerable, implying that early training dynamics are already shaped by additional factors such as stochastic policy-gradient noise, exploration, and architecture-dependent effects. 

The right plot provides a fairly definitive answer to how far this diagnostic indicator can go. Specifically, correlation between the prescreening score and the final-stage stability index $S_{\text{final}}$ is very weak, with the Pearson coefficient in particular remaining close to zero. In addition, the variance is substantial, making it difficult to argue that any context exhibits a trend as clear or robust as that shown in the left-hand plot. This result is also broadly expected. Predicting the stable success of the entire learning process solely from the absence of early-stage gradient information is inherently difficult because many other factors contribute both quantitatively and qualitatively to training outcomes, including the uncertainty intrinsic to stochastic gradient optimizers, the structural properties inherent to PQCs, and the extent to which the training dynamics explore regions that are amenable to approximation (expressivity).

Taken together, $\Gamma_\epsilon$ is most useful as a one-sided `initialization' diagnostic rather than a standalone calibrated predictor of final training quality: it is strongly aligned with initialization survival, only partially informative for early transients, and hardly informative for final stability.

\subsection{Sensitivity to the Bin Count $B$}
Finally, because the bin count is a crucial hyperparameter for MI-TET estimation, we investigate how varying the discretization resolution affects (i) the MI-TET magnitude itself and (ii) the two terms appearing in the trainability theorem. Fig.~6 summarizes this sensitivity analysis. The first panel shows how the mean MI-TET value changes as a function of the bin count. The second panel (bottom-left) reports how each of the RHS components in the trainability bound varies with the bin count, and the third panel (bottom-right) provides a zoomed-in view of $\text{RHS}_{\text{main}}$ from the second panel. 

\begin{figure*}
    \centering
    \includegraphics[width=\textwidth]{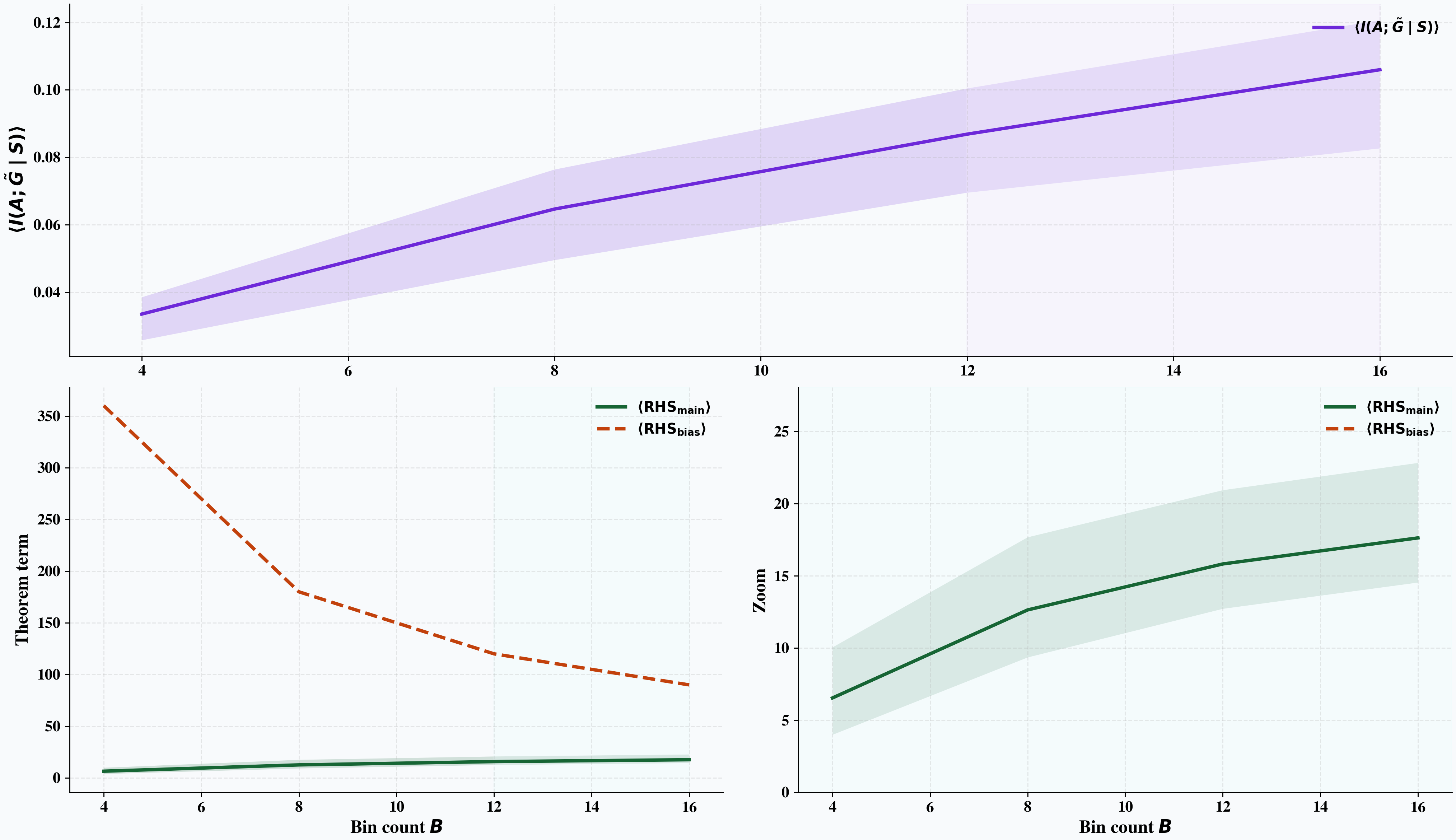}
    \caption{
    \textbf{Bin-count sensitivity of MI-TET and the trainability-bound components}.
    \textbf{Top:} mean MI-TET as a function of the discretization bin count \(B\).
    \textbf{Bottom-left:} bin-count dependence of the RHS components of the trainability theorem (including \(\mathrm{RHS}_{\mathrm{main}}\) and the discretization/bias term \(\mathrm{RHS}_{\mathrm{bias}}\propto \Delta\)).
    \textbf{Bottom-right:} zoomed-in view of \(\mathrm{RHS}_{\mathrm{main}}\).
    Increasing \(B\) improves discretization resolution but can induce sparsity effects at large \(B\). Correspondingly, \(\mathrm{RHS}_{\mathrm{main}}\) tracks the MI-TET trend, whereas \(\mathrm{RHS}_{\mathrm{bias}}\) decreases with \(B\) as \(\Delta\) shrinks.
    }
    \label{fig:bin_sensitivity}
\end{figure*}

Overall, the observed bin-count dependence is consistent with the qualitative behavior anticipated from the theory. Since the bin count directly controls the discretization resolution used to form MI-TET, increasing the bin count effectively increases the representational `sharpness' with which MI-TET can be resolved from data, and this is reflected in the first panel. At the same time, the marginal gain in MI-TET diminishes as the bin count increases: beyond a certain point, overly fine binning induces a data-sparsity regime in which empirical histograms become noisy and under-populated, which can in turn suppress the estimated MI-TET. Thus, the first panel suggests an implicit tradeoff--higher resolution is beneficial up to a point, after which sparsity effects can counteract the expected increase.

The trainability-bound components exhibit complementary bin-count trends that make the tradeoff explicit. Because $\text{RHS}_{\text{main}}$ contains the MI-TET factor, it qualitatively tracks the MI-TET trajectory as the bin count varies. In contrast, the discretization/bias term $\text{RHS}_{\text{bias}} \propto \Delta$ decreases as the bin count increases, since $\Delta$ (the discretization width) shrinks under finer binning. Consequently, the trainability theorem's RHS exhibits a clear bin count tradeoff between its two dominant contributions: increasing the bin count tends to increase the MI-TET-driven main term while simultaneously reducing the discretization/bias term. This highlights that the bin count is not merely a numerical detail but a meaningful hyperparameter that mediates competing effects in both MI-TET estimation and the associated theoretical bound. 

\section{Conclusion and Future Works}
\subsection{Summary and Value of Findings}
Throughout the whole paper, we mainly introduce MI-TET, a mutual information-based diagnostic tailored to policy gradient with PQC policies, with the goal of quantifying both temporal expressivity and trainability in a way that respects RL's inherently time-varying exploration-exploitation dynamics. MI-TET is defined as the conditional mutual information between the action and a discretized reward signal, $\operatorname{MI-TET} := I(A; \Tilde{Y} | \bar S)$ (or $I(A; \Tilde{Y} | S)$) with $Y \in \{G_t, Q_t^{\pi_\theta}\}$, where binning yields a computationally robust estimator that avoids continuous probability density estimation overhead. In parallel, we formulate expressivity for RL as a \emph{temporal volatility} notion, rather than following conventional \emph{capacity-like} perspective. In this sense, our new temporal expressivity notion signifies the degree to which the policy's action distribution changes across recent sampling indices, captured by weighted Jensen-Shannon Divergence and equivalently $\operatorname{Expr} = I(A; Z | S)$. We then establish information-theoretic inequalities linking these quantities: MI-TET upper-bounds (up to design-fixed terms) the scaled gradient norm, providing an interpretable proxy for trainability via the dominant time-varying factor $\sigma_{g | \bar S} \sqrt{\operatorname{MI-TET}}$, and MI-TET also controls temporal expressivity through $\operatorname{Expr} \leq \operatorname{MI-TET} + I(A; Z | \tilde{Y}, S)$.

Empirically, in a CartPole-v1 REINFORCE pipeline with softmax-PQC policies, MI-TET exhibits a characteristic transition over the course of learning: it initially increases as the agent actively explores the action-reward dependence structure, but subsequently declines once training enter a more stable regime, where reduced policy entropy naturally suppresses the mutual information-based signal, MI-TET. While the full trainability bound can be numerically loose due to the discretization/bias term, the core multiplicative factor co-moves with overall temporal changes in the observed gradient norm, supporting MI-TET's practical value as an online diagnostic rather than a tight certificate. The expressivity inequality is cleanly satisfied throughout training, and the residual term is substantial in early/mid phases but tends to diminish later, indicating that local stationarity emerges gradually rather holding globally. Finally, we show that MI-TET can be leveraged at initialization to derive assumption-based, one-sided prescreening score that helps eliminate initialization-fragile PQC architectures, and we also validate the expected bin-count trade-off: finer binning reduces discretization bias while risking sparsity/noise in MI estimates. 

Overall, MI-TET offers an RL-native, online-trackable, information-theoretic handle on both temporal expressivity and gradient survival, with direct implications for monitoring learning dynamics, tradeoff between trainability and temporal expressivity, and for architecture selection in quantum policy gradient pipelines.

\subsection{Limitations and Future Directions}
A first limitation of the present work is that our theoretical guarantees are predominately upper-bound statements. While the derived inequalities justify MI-TET as a meaningful proxy that can \emph{control} both gradient magnitude (trainability) and temporal expressivity, they do not imply that the target quantities must co-vary tightly with MI-TET along the whole phases of learning. A natural direction is then to pursue complementary lower-bound results: if one can show that mutual information--motivated by the same principle--also provides a nontrivial lower bound (perhaps up to controllable scaling factor or residual terms) on  trainability/expressivity-related quantities, the resulting theory would become substantially more complete and predictive. 

Moreover, the current prescreening score construction is derived from relatively simple probabilistic proposition and relies on strong concentration-type assumptions on initialization-time MI-TET. Strengthening the theoretical foundations of this score remains open: promising routes may include weakening the tail/concentration assumptions or incorporating PQC structure-specific features into its mathematical formulation etc.

A second major limitation is empirical: due to the well-known bottlenecks of current quantum backends and the cost of repeated quantum circuit evaluations, our experiments were necessarily confined to a relatively simple benchmark and a restricted set of PQC configurations. It remains to be tested whether the qualitative behaviors observed here--e.g., MI-TET transitions across exploration-to-exploitation regimes, the co-movement of the dominant trainability factors, and the role of the residual term in expressivity--persist under harder control tasks, richer observation spaces, longer horizons, different exploration schedules, and more diverse PQC architectures (including deeper circuits and alternative measurement schemes). 

Beyond these extensions, a particularly ambitious open direction is to `quantize' MI-TET itself. While our current MI-TET is a purely classical statistic estimated from trajectories, a genuinely quantum RL setting would allow the action register, the environment interaction, and the reward encoding to be all treated as quantum systems. This motivates defining a quantum MI-TET via quantum mutual information, e.g.,
\begin{align}
    & \operatorname{qMI\text{-}TET} := \mathbb{E}_s\left[I_q(A : \tilde{Y})_{\rho^{(s)}_{A\tilde{Y}}}\right] \\
    & I_q(A : \tilde{Y}) = S(\rho_A) + S(\rho_{\tilde{Y}}) - S(\rho_{A\tilde{Y}}),
\end{align}
where $\rho_{A\tilde{Y}}^{(s)}$ is the reduced state obtained by coherently coupling the policy's action register with a reward encoding register under a fixed state $s$. Estimating these von Neumann entropies on hardware is nontrivial, but recent neural estimators such as QMINE \cite{QMINE}--based on the quantum Donsker-Varadhan representation--provide a variational route to learn $S(\rho)$ (and hence $I_q$). Embedding the variational policy and a reward-bearing environment into a joint quantum state and training a QMINE-style estimator alongside policy optimization could therefore enable monitoring trainability and temporal expressivity directly and naturally in the quantum domain.

Another promising avenue is to couple MI-TET with quantum resource measures, thereby moving toward namely resource-aware quantum reinforcement learning. In distributed quantum RL scenarios, where the agent and environment reside a separate quantum nodes connected by constrained quantum channels, an effective policy should not only achieve high expected reward but also respect limitations on communication and entanglement \cite{HaldarRepeatersRL2024}. Quantum uncommon information--operationally define as the minimal quantum communication cost required to exchange two quantum states \cite{OppenheimWinterUncommon}--provides a natural notion of such a cost, and recent work has proposed QDVR-based neural estimators for upper and lower bounds on this quantity \cite{QUINE}. Combining these tools, one can envision a multi-objective quantum RL framework where MI-TET serves as a information-theoretic indicator of learnability and temporal expressivity, while neural estimates of quantum uncommon information quantify the communication resources consumed by the policy.

\section*{Acknowledgement}
Our special thanks go to Junghee Ryu (Korean Institute of Science and Technology Information, KISTI) and Harris Junseo Lee (Seoul National University) for their generous feedback and critical comments. Their contributions were instrumental in refining the core ideas and enhancing the overall depth of this work. This work was supported by the National Research Foundation of Korea (NRF) through a grant funded by the Ministry of Science and ICT (Grants Nos. RS-2025-00515537 and RS-2023-00211817), the Institute for Information \& Communications Technology Promotion (IITP) grant funded by the Korean government (MSIP) (Grants Nos. RS-2025-02304540 and RS-2019-II190003), the National Research Council of Science \& Technology (NST) (Grant No. GTL25011-000), and the Korea Institute of Science and Technology Information (Grant No. P25026). We acknowledge the Yonsei University Quantum Computing Project Group for providing support and access to the Quantum System One (Eagle Processor), which is operated at Yonsei University. \\

\section*{Data availability}
The code and data used in this work are available at the following GitHub link: \url{https://github.com/absolute-injury/MI-TET/tree/main}
\label{githubLink}

\section*{Declaration of competing interest}
The authors declare that they have no known competing financial interests or personal relationships that could have appeared to influence the work reported in this paper. \\

\section*{Author Contribution}
J.J. contributed to this work, undertaking the primary responsibilities, including the development of the main ideas, mathematical proofs, initial drafting, and revisions of the paper. D.J., and K.J. provided valuable feedback in shaping the core ideas and overall direction of the work, while K.J. also supervised the whole research. All authors discussed the results and contributed to the final paper.

%%%%%
\bibliographystyle{unsrt}
\bibliography{references}
%\bibliography{GZP_references}

\clearpage
\appendix
\crefalias{section}{appendix}
\begin{center}
    \textbf{\large Appendix}
\end{center}
%\setcounter{equation}{0}
%\setcounter{figure}{0}
%\setcounter{table}{0}
%\makeatletter
%\renewcommand{\theequation}{S\arabic{equation}}
%\renewcommand{\thefigure}{S\arabic{figure}}
%\renewcommand{\bibnumfmt}[1]{[S#1]}

\section{Proof for the Trainability Theorems}
\label{Appendix A}

First of all, we prove our trainability theorem 1. If we fix state to be $s$ then our objective function becomes
\begin{equation}
    \eta_s(\theta) = \sum_a \pi_\theta(a) R(a).
\end{equation}
Then
\begin{equation}
    g_s(\theta) = \nabla_\theta \eta_s(\theta) = \sum_a \nabla_\theta \pi_\theta(a) R(a).
\end{equation}
Here,
\begin{equation}
    \nabla_\theta \pi_\theta(a) = \pi_\theta(a) \nabla_\theta \log \pi_\theta(a) = \pi_\theta(a) S_\theta(a),
\end{equation}
making
\begin{equation}
    g_s(\theta) = \sum_a \pi_\theta(a) S_\theta(a) R(a) = \mathbb{E}\left[S_\theta(A) R(A)\right].
\end{equation}

Next, we show the expected value of score function is equal to zero.
\begin{align}
    \mathbb{E}\left[S_\theta(A)\right] & = \sum_a \pi_\theta(a) \nabla_\theta \log \pi_\theta(a) = \sum_a \nabla_\theta \pi_\theta(a) \\
    & = \nabla_\theta \sum_a \pi_\theta(a) = \nabla_\theta 1 = 0
\end{align}
If we set $m := \mathbb{E}\left[R(A)\right]$ (moreover, let us denote $R := R(A)$ for simplicity from now on),
\begin{align}
    \mathbb{E}\left[S_\theta(A)(R - m)\right] & = \mathbb{E}\left[S_\theta(A) R\right] - m\mathbb{E}\left[S_\theta(A)\right] \\
    & = \mathbb{E}\left[S_\theta(A) R\right] \\
    & = g_s(\theta).
\end{align}

Let us define reward-conditional mean score as follows.
\begin{align}
    \Delta (r) & := \mathbb{E}\left[S_\theta(A) | R = r\right] - \mathbb{E}\left[S_\theta(A)\right] \\
    & = \mathbb{E}\left[S_\theta(A) | R = r\right].
\end{align}
Then, due to the law of total expectations,
\begin{align}
    g_s(\theta) & = \mathbb{E}\left[\mathbb{E}\left[S_\theta(A) (R - m) | R\right]\right] \\
    & = \mathbb{E}\left[(R - m) \mathbb{E}\left[S_\theta(A) | R\right]\right] = \mathbb{E}\left[(R - m) \Delta (R)\right].
\end{align}
According to the Jensen's Inequality,
\begin{equation}
    ||g_s(\theta)|| \leq \mathbb{E}\left[|R - m| ||\Delta (R)||\right].
\end{equation}

Moreover, let $P_{A | R = r}$, $P_A$ denote the conditional probability distribution and the marginal action distribution, respectively.
Then 
\begin{align}
    \Delta (r) & = \mathbb{E}\left[S_\theta(A) | R = r\right] - \mathbb{E}\left[S_\theta(A)\right] \\
    & = \sum_a (P(A = a | R = r) - P(A = a)) S_\theta(a).
\end{align}
We already know $||S_\theta(a)|| \leq G_{\max}$ holds due to the assumption. Therefore,
\begin{align}
    ||\Delta (r)|| & \leq \sum_a |P(A = a | R = r) - P(A = a)| ||S_\theta(a)|| \\
    & \leq G_{\max} \sum_a |P(A = a | R = r) - P(A = a)|.
\end{align}
Total Variation, which is defined as
\begin{equation}
    \operatorname{TV}(P, Q) := \frac{1}{2} \sum_x |P(x) - Q(x)|
\end{equation}
can be expressed as 
\begin{equation}
    \sum_a |P(A = a | R = r) - P(A = a)| = 2\operatorname{TV}(P_{A | R = r}, P_A)
\end{equation}
for this case. This makes
\begin{equation}
    ||\Delta (r)|| \leq 2G_{\max} \operatorname{TV}(P_{A | R = r}, P_A),
\end{equation}
thus
\begin{equation}
    ||g_s(\theta)|| \leq 2G_{\max} \mathbb{E}\left[|R - m| \operatorname{TV}(P_{A | R}, P_A)\right].
\end{equation}

Now we apply expectation-type Cauchy-Schwarz inequality
\begin{equation}
    \mathbb{E}\left[XY\right] \leq \sqrt{\mathbb{E}\left[X^2\right]} \sqrt{\mathbb{E}\left[Y^2\right]},
\end{equation}
as $X = |R - m|$ and $Y = \operatorname{TV}(P_{A | R}, P_A)$.
This will give us
\begin{equation}
    ||g_s(\theta)|| \leq 2G_{\max} \sqrt{\operatorname{Var}(R)} \sqrt{\mathbb{E}\left[\operatorname{TV}(P_{A | R}, P_A)^2\right]}.
\end{equation}

Finally, we apply Pinsker Inequality
\begin{equation}
    \operatorname{TV}(P, Q)^2 \leq \frac{1}{2}D_{KL}(P || Q).
\end{equation}
Then
\begin{equation}
    \mathbb{E}\left[\operatorname{TV}(P_{A | R}, P_A)^2\right] \leq \frac{1}{2}\mathbb{E}\left[D_{KL} (P_{A | R} || P_A)\right] = \frac{1}{2} I(A; R).
\end{equation}
Now our proof is complete. \qed

Next is the proof of the trainability theorem 2. Let us introduce $\widehat{R} := g(\tilde{R})$ and $\epsilon := R - \widehat{R}$, which satisfy $R = \widehat{R} + \epsilon$. Here $\widehat{R}$ denotes the midpoint of each bin, thus
\begin{equation}
    |\epsilon| = |R - g(\tilde{R})| \leq \frac{\Delta}{2}
\end{equation}
naturally holds.

Moreover, we can decompose the gradient term as
\begin{align}
    g_s(\theta) & = \nabla_\theta \mathbb{E}\left[R\right] \\
    & = \nabla_\theta \mathbb{E} \left[\widehat{R}\right] + \nabla_\theta \mathbb{E}\left[\epsilon\right]
\end{align}
making 
\begin{equation}
    ||g_s(\theta)|| \leq \left\| \nabla_\theta \mathbb{E}\left[\widehat{R} \right]\right\| + ||\nabla_\theta \mathbb{E}\left[\epsilon\right]||.
\end{equation}

If we view $\widehat{R} = g(\tilde{R})$ as the new reward-like value, then we can apply theorem 1 to this term. This is validated through the fact that $\widehat{R}$ is just a relabeling of $\tilde{R}$, with mapping function $g$. Thus
\begin{align}
    \left\|\nabla_\theta \mathbb{E}\left[\widehat{R}\right]\right\| & \leq \sqrt{2}G_{\max} \sqrt{\operatorname{Var}(\widehat{R})} \sqrt{I(A; \widehat{R})} \\
    & = \sqrt{2}G_{\max} \sqrt{\operatorname{Var}(g(\tilde{R}))} \sqrt{I(A; \tilde{R})}.
\end{align}

For the second term, by following the exact same logic as in the previous proof, we get the relation
\begin{equation}
    \nabla_\theta \mathbb{E}\left[\epsilon\right] = \mathbb{E}\left[S_\theta(A) \epsilon\right].
\end{equation}
Hence, 
\begin{align}
    ||\nabla_\theta \mathbb{E}\left[\epsilon\right]|| & \leq \mathbb{E}\left[||S_\theta(A)|| |\epsilon|\right] \\
    & \leq G_{\max} \mathbb{E}\left[|\epsilon|\right] \leq G_{\max} \frac{\Delta}{2}
\end{align}

By combining those two terms, the trainability theorem 2 is now proved.

Finally, we formalize the detailed proof structure for the main trainabilty theorem 3.
Write
\begin{equation}
    \alpha_T(\gamma) := \left(\sum_{t=0}^{T - 1} \gamma^t\right)^{-1}
\end{equation}
and 
\begin{equation}
    \eta'(\theta) := \alpha_T(\gamma) \eta(\theta).
\end{equation}
Equivalently, for $\gamma \neq 1$,
\begin{equation}
    \alpha_T(\gamma) = \frac{1 - \gamma}{1 - \gamma^T},
\end{equation}
while for $\gamma = 1$,
\begin{equation}
    \alpha_T(1) = \frac{1}{T}.
\end{equation}

From the finite-horizon policy gradient identity, using the return-to-go notation
\[
Y_t\in\{G_t,Q_t^{\pi_\theta}(S_t,A_t)\},
\]
we have
\[
\nabla_\theta \eta'(\theta)
=
\alpha_T(\gamma)
\sum_{t=0}^{T-1}\gamma^t\,
\mathbb E\!\left[
\nabla_\theta \log \pi_\theta(A_t\mid S_t)\,Y_t
\right].
\]
Define
\[
c_t:=\alpha_T(\gamma)\gamma^t,\qquad t=0,\dots,T-1.
\]
Since
\[
\sum_{t=0}^{T-1} c_t
=
\alpha_T(\gamma)\sum_{t=0}^{T-1}\gamma^t
=1,
\]
the sequence $(c_t)_{t=0}^{T-1}$ is a probability mass function.

Let $J$ be an auxiliary random index, independent of the sampled trajectory, with
\[
\Pr(J=t)=c_t.
\]
Now define
\[
\bar S :=(J,S_J),\quad A:=A_J,\quad Y:=Y_J,\quad \tilde{Y} := \tilde{Y}_J.
\]
Then, by conditioning on the event $\{J=t\}$,
\[
\mathbb E\!\left[
\nabla_\theta \log \pi_\theta(A\mid \bar S)\,Y
\right]
=
\sum_{t=0}^{T-1} c_t\,
\mathbb E\!\left[
\nabla_\theta \log \pi_\theta(A_t\mid S_t)\,Y_t
\right].
\]
Using $c_t=\alpha_T(\gamma)\gamma^t$, we obtain
\[
\nabla_\theta \eta'(\theta)
=
\mathbb E\!\left[
\nabla_\theta \log \pi_\theta(A\mid \bar S)\,Y
\right].
\]

For each fixed augmented state $\bar s$, define
\[
g_{\bar s}(\theta)
:=
\mathbb E\!\left[
\nabla_\theta \log \pi_\theta(A\mid \bar S)\,Y
\mid \bar S=\bar s
\right].
\]
By the law of total expectations,
\begin{equation}
    \nabla_\theta \eta'(\theta) = \mathbb{E}_{\bar S}\left[g_{\bar S}(\theta)\right]
\end{equation}
Hence,
\begin{equation}
    ||\nabla_\theta \eta'(\theta)|| \leq \mathbb{E}_{\bar S}\left[||g_{\bar S}(\theta)||\right]
\end{equation}

Now fix $\bar s$. Conditional on $\bar S = \bar s$, the problem reduces to the one-shot setting of Theorem 2 with reward variable $Y\mid(\bar S = \bar s)$. Therefore,
\begin{align}
    ||g_s(\theta)|| & \leq \sqrt{2}G_{\max} \sqrt{\operatorname{Var}(g(\tilde{Y}) | \bar S = \bar s)}\sqrt{I(A; \tilde{Y} | \bar S = \bar s)} \\
    & + G_{\max} \frac{\Delta}{2}.
\end{align}
Taking expectation with respect to $\bar S$ gives
\begin{align}
    ||\nabla_\theta \eta'(\theta)|| & \leq \sqrt{2}G_{\max} \mathbb{E}_{\bar S}\left[\sqrt{\operatorname{Var}(g(\tilde{Y}) | \bar S = \bar s)} \sqrt{I(A; \tilde{Y} | \bar S = \bar s)}\right] \\
    & + G_{\max} \frac{\Delta}{2}.
\end{align}

Applying the Cauchy--Schwarz inequality with
\[
X(\bar S):=\sqrt{\operatorname{Var}(g(\widetilde Y)\mid \bar S)},
\qquad
Z(\bar S):=\sqrt{I(A;\widetilde Y\mid \bar S)},
\]
we get
\[
\mathbb E_{\bar S}[X(\bar S)Z(\bar S)]
\le
\sqrt{\mathbb E_{\bar S}[X(\bar S)^2]}\,
\sqrt{\mathbb E_{\bar S}[Z(\bar S)^2]}.
\]
That is,
\begin{align}
    & \mathbb E_{\bar S}\!\left[
\sqrt{\operatorname{Var}(g(\widetilde Y)\mid \bar S)}
\sqrt{I(A;\widetilde Y\mid \bar S)}
\right] \\
 & \le
\sqrt{\mathbb E_{\bar S}[\operatorname{Var}(g(\widetilde Y)\mid \bar S)]}
\sqrt{\mathbb E_{\bar S}[I(A;\widetilde Y\mid \bar S)]}.
\end{align}
Therefore,
\begin{align}
    & \|\nabla_\theta \eta'(\theta)\| \\
& \le
\sqrt{2}\,G_{\max}
\sqrt{\mathbb E_{\bar S}[\operatorname{Var}(g(\widetilde Y)\mid \bar S)]}
\sqrt{\mathbb E_{\bar S}[I(A;\widetilde Y\mid \bar S)]} \\
& + G_{\max} \frac{\Delta}{2}.
\end{align}
Finally, by the definition of conditional mutual information,
\[
\mathbb E_{\bar S}[I(A;\widetilde Y\mid \bar S)]
=
I(A;\widetilde Y\mid \bar S).
\]
Hence
\begin{align}
    & \|\nabla_\theta \eta'(\theta)\| \\
& \le
\sqrt{2}\,G_{\max}
\sqrt{\mathbb E_{\bar s}[\operatorname{Var}(g(\widetilde Y)\mid \bar S=\bar s)]}
\sqrt{I(A;\widetilde Y\mid \bar S)} \\
& + G_{\max}\frac{\Delta}{2}.
\end{align}
This proves Theorem 3.
\qed

\section{Practical Estimation of the Initialization-time Prescreening Score}
\label{Appendix B}

We now describe how the initialization-time prescreening score can be estimated in practice under a fixed task and evaluation protocol. For sampled initial parameters $\theta_1, \cdots, \theta_n \stackrel{\mathrm{i.i.d.}}{\sim} \mathcal{D}_{\mathrm{init}}$, we define
\begin{align}
    \widehat{\Gamma}_\epsilon 
    &:= \left[\widehat{\tau}_\epsilon - \widehat{\mu}_M\right]_+ \\
    &= \left[\left(\frac{\epsilon - b}{a\sqrt{\widehat{V}}}\right)^2 - \frac{1}{n}\sum_{i=1}^n M(\theta_i)\right]_+,
\end{align}
where
\begin{equation}
    \widehat{\mu}_M := \frac{1}{n}\sum_{i=1}^n M(\theta_i).
\end{equation}

The empirical mean term $\widehat{\mu}_M$ is directly obtained from the sampled initializations, so the remaining task is to specify a practical approximation for the threshold term $\widehat{\tau}_\epsilon$. To this end, let $\widehat{V}$ denote the empirical 90\% (or 95\%) quantile of the sampled variance values. This quantile-based choice serves as a conservative empirical surrogate for the deterministic variance envelope $\bar{V}$ used in the theoretical bound.

Under this approximation, the threshold term is given by
\begin{equation}
    \widehat{\tau}_\epsilon
    =
    \left(\frac{\epsilon - b}{a\sqrt{\widehat{V}}}\right)^2.
\end{equation}

Recall that $a=\sqrt{2}G_{\max}$ and $b=G_{\max}\Delta/2$, so that the practical computation of $\widehat{\tau}_\epsilon$ depends on the quantities $\epsilon$, $G_{\max}$, $\Delta$, and $\widehat{V}$. Here, the threshold $\epsilon$ is a user-specified parameter that determines the minimum gradient magnitude regarded as operationally meaningful. Since the reward signal $Y$ is discretized by design, the discretization error $\Delta$ is explicitly determined by the chosen binning scheme. Moreover, under a clipped evaluation or training protocol, $G_{\max}$ can be fixed in advance by the prescribed clipping constant $c$. Therefore, all components constituting $\widehat{\tau}_\epsilon$ are either empirically estimable from sampled initializations or explicitly fixed by the evaluation protocol. Consequently, the plug-in score $\widehat{\Gamma}_\epsilon$ is practically computable under a common task setting. \\

\section{Full Correlation Matrix}
\label{Appendix C}

\begin{figure*}
    \centering
    \includegraphics[width=\textwidth]{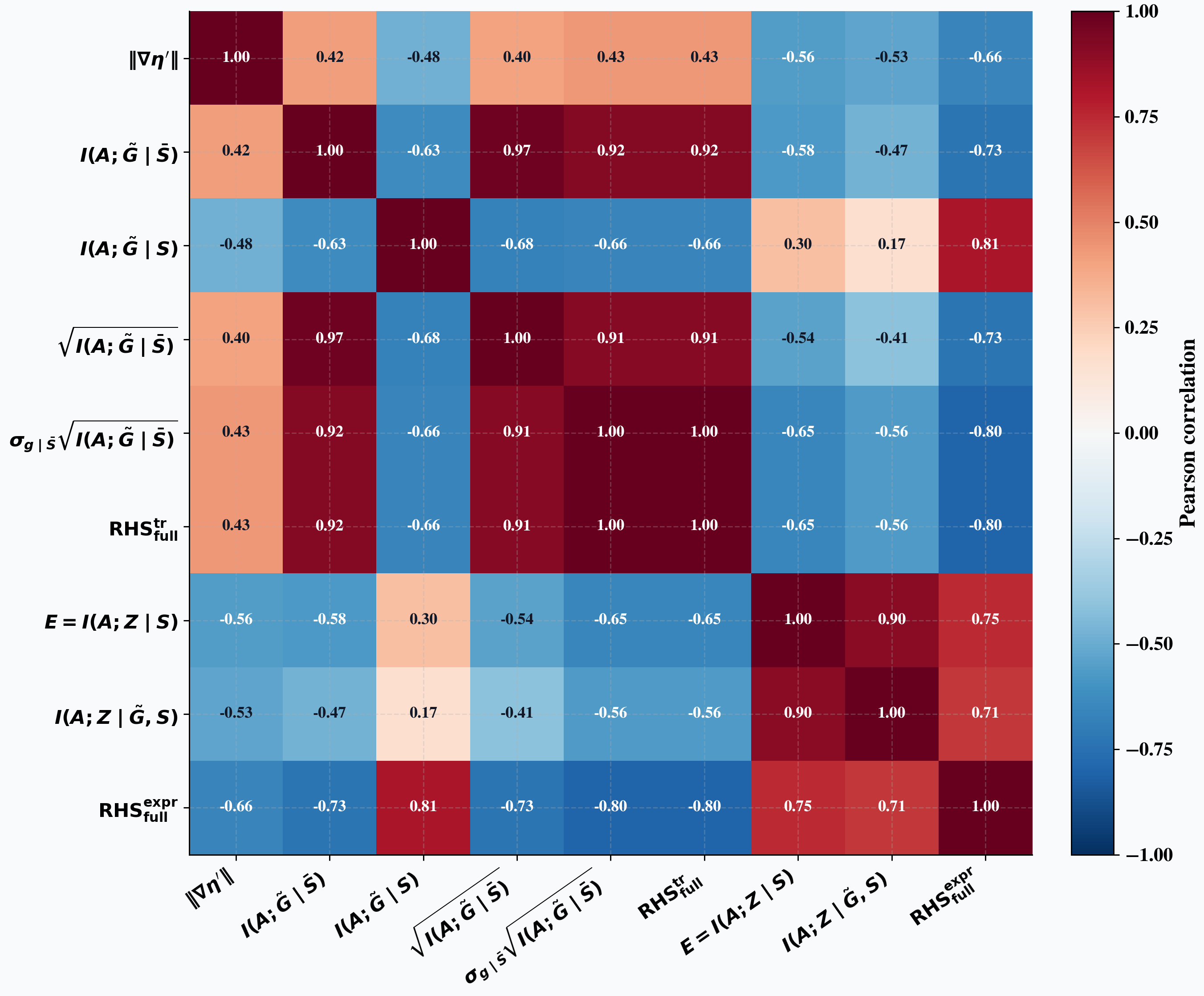}
    \caption{\textbf{Full Correlation Matrix between Various Learning-Related Values}.}
    \label{fig:correlation matrix}
\end{figure*}

\end{document}